\documentclass[preprint]{elsarticle}

\usepackage{hyperref}

\journal{Journal of Computational Physics}

\usepackage{amsmath}
\usepackage{amssymb}
\usepackage{amsthm}
\usepackage{subcaption}
\usepackage{stmaryrd} 
\usepackage{cleveref}
\usepackage{hhline}
\usepackage{multirow}
\newtheorem{lemma}{Lemma}[section]
\newtheorem{definition}{Definition}[section]
\newtheorem{theorem}{Theorem}[section]
\newtheorem{corollary}{Corollary}[lemma]

\usepackage{xcolor}

\usepackage{ulem}
\usepackage{inputenc}

\newcommand{\vect}[1]{\boldsymbol{#1}}
\newcommand{\matr}[1]{\mathbf{#1}}
\newcommand{\mat}[1]{\overleftrightarrow{#1}}

\def\bhat{{\hat b}}
\def\Bmag{B}

\def\Vol{V}

\def\Vnabla{\vec \nabla}

\def\VnablaPerp{\vec \nabla_\perp}
\def\VnablaPara{\vec \nabla_{||}}

\newcommand\abs[1]{\left|#1\right|}
\newcommand\norm[1]{\left|#1\right|}
\newcommand{\pdv}[2]{{\partial_{#2} #1}}

\newcommand\avg[1]{\left< #1\right>}

\def\zhat{\hat z}
\def\fvor{w}
\def\fpot{\varphi}
\def\vDrift{\vec{v}}
\def\vDriftxy{\vDrift}

\newcommand{\jumpOp}[1]{\left\llbracket #1 \right\rrbracket }
\newcommand{\avgOp}[1]{\left\{ #1 \right\}}
\newcommand{\PoissonBracket}[2]{\left\{\left\{#1 , #2\right\}\right\} }

\def\ftestLtwo{\eta}
\def\ftestHone{\psi}
\def\ftestgen{\ftestHone}

\def\tauh{\tau_h}

\newcommand{\Ltwo}[1]{L^2_{#1}(\tauh)}
\newcommand{\LtwoVec}[1]{\left[L^2_{#1}(\tauh)\right]^d}
\newcommand{\Hone}[1]{H^1_{#1}(\tauh)}
\newcommand{\Hcurl}[1]{H^{\rm Curl}_{#1}(\tauh)}
\newcommand{\Hdiv}[1]{H^{\rm Div}_{#1}(\tauh)}

\def\ppot{p}
\def\qvor{\tilde{p}}

\def\Honep{\Hone{\ppot}}
\def\Honeq{\Hone{\qvor}}
\def\Ltwoq{\Ltwo{\qvor}}

\def\VB{\vec{B}}

\def\VE{\vec{E}}
\def\VF{\vec{F}}
\def\VJ{\vec{J}}
\def\Vv{\vec{v}}

\def\CO{{\mathcal O}}
\def\Energy{{\mathcal E}}
\def\Enstrophy{{\mathcal S}}

\def\DiffusionMatrix{\matr{L}}
\def\DiffusionMatrixElements{L}

\def\OverlapMatrix{\matr{N}}

\usepackage{setspace}

\begin{document}

\begin{frontmatter}

\title{ Arbitrary Order Energy and Enstrophy Conserving Finite Element Methods for 2D Incompressible Fluid Dynamics and Drift-Reduced Magnetohydrodynamics}
\tnotetext[mytitlenote]{Energy and Enstrophy Conserving FEM scheme }



\author[LLNL]{Milan Holec}
\cortext[mycorrespondingauthor]{Corresponding author}
\ead{holec1@llnl.gov}
\ead[url]{mfem.org}

\author[LLNL]{Ben Zhu}

\author[LLNL]{Ilon Joseph}

\author[LLNL]{Christopher J. Vogl}

\author[LANL]{Ben S. Southworth}

\author[LLNL]{Alejandro Campos}

\author[LLNL]{Andris M. Dimits}

\author[LLNL]{Will E. Pazner}

\address[LLNL]{Lawrence Livermore National Laboratory, P.O. Box 808, Livermore, CA 94551, U.S.}
\address[LANL]{Los Alamos National Laboratory, P.O. Box 1663, Los Alamos, NM 87545 U.S.}

\begin{abstract}
Maintaining conservation laws in the fully discrete setting is critical 
for accurate long-time behavior of numerical simulations and requires accounting for discrete
conservation properties in both space and time. This paper derives arbitrary order finite element exterior
calculus spatial discretizations for the two-dimensional (2D) Navier-Stokes and drift-reduced magnetohydrodynamic
equations that conserve both energy and enstrophy to machine precision when coupled with generally symplectic
time-integration methods. Both continuous and discontinuous-Galerkin (DG) weak formulations can ensure conservation,
but only generally symplectic time integration methods, such as the implicit midpoint method, permit exact conservation in time.
Moreover, the symplectic implicit midpoint method yields an order of magnitude speedup over explicit schemes.
The methods are implemented using the MFEM library and the solutions are verified for an extensive suite of 2D neutral fluid turbulence test problems.
Numerical solutions are verified via comparison to a semi-analytic linear eigensolver as well as to the finite difference Global Drift Ballooning (GDB) code. 
However, it is found that turbulent simulations that conserve both energy and enstrophy tend to have too much power at high wavenumber and that this part of the spectrum should be controlled by reintroducing artificial dissipation.
The DG formulation  allows upwinding of the advection operator which dissipates enstrophy while still maintaining conservation of energy.
Coupling upwinded DG with implicit symplectic integration appears to offer the best compromise of allowing mid-range
wavenumbers to reach the appropriate amplitude while still controlling the high-wavenumber part of the spectrum. 
\end{abstract}


\end{frontmatter}

%
\section{Introduction}

In the absence of dissipation, the equations of fluid dynamics and plasma physics possess conserved quantities, such as energy and enstrophy (in 2D) or helicity (in 3D), that are invariant under  evolution in time.
Many numerical methods, including finite difference, finite element, finite volume, and pseudo-spectral methods, have  been explored for the purpose of developing numerical approximations for solutions of the  underlying partial differential equations, such as the Navier-Stokes equations, the magnetohydrodynamics (MHD) equations, and the  Braginskii equations.
It is important for any numerical approximation scheme   to approximately conserve these   invariants  in order to ensure the long-term accuracy (and often the stability) of the approximate solution.
The introduction of finite element exterior calculus (FEEC) \cite{ArnoldUniDGSIAM2002, Arnold2006, Arnold2010} has reinvigorated interest in finite element approaches to fluid and plasma physics and, as will be shown, provides a natural framework for enforcing these conservation laws.

For the equations of fluid and plasma physics, the numerical discretization of the advection operator is key to ensuring conservation.
In 2D, the finite difference community benefits from the ability to write the advection operator using the numerical Arakawa bracket \cite{Arakawa1966jcp} to conserve the two quadratic invariants: energy and enstrophy.
In contrast, existing finite element approaches tend to introduce more dissipation than necessary, leading to an artificial reduction of energy at small spatial scales.
One approach to enforcing conservation focuses on using an antisymmetric form for the advection operator \cite{Morinishi98jcp, Morinishi10jcp, Halpern18pop, Halpern20pop}.
Another conservative approach \cite{Einkemmer2014}, developed  for a specific discontinuous Galerkin (DG)  finite element method (FEM), expresses the advection operator in the form of a Poisson bracket that enforces both antisymmetry as well as a discrete divergence theorem for  products of functions.
A review of structure-preserving integration techniques used in computational engineering was recently given in \cite{Sharma2020}.

Recently, low-order energy and enstrophy conserving schemes were developed within the FEEC framework for the 2D Navier-Stokes equations \cite{NataleQJRMS2017}, \cite{NataleIMA2018}, as well as for the rotating shallow water equations and the shallow water equations on the sphere \cite{CotterJCP2012, McRaeQJRMS2014,   BauerJCP2018}.
In this work, an FEEC approach to the 2D incompressible Navier-Stokes \cite{Kraichnan1980rpp} and the 2D drift-reduced magnetohydrodynamics (MHD) equations \cite{ Hasegawa1978Mima,Horton1994,Hazeltine2003Meiss} is derived that conserves both energy and enstrophy  for arbitrary polynomial order, in the absence of dissipation.
Because the theorems of vector calculus are satisfied at the discrete level within the FEEC framework, the spatial discretization automatically conserves both energy and enstrophy as long as the proper finite element spaces are chosen for the primary variables.

A conservative scheme also requires an appropriate choice of numerical time integration technique,
an important fact that has not been discussed by previous authors.
This is particularly relevant here, as both energy and enstrophy correspond to quadratic functions, and significant care must be taken to conserve quadratic invariants (in contrast to other physical problems where energy is a linear first integral invariant that is trivial to conserve using standard time integration schemes).
In fact, Ref. \cite{Chartier.2006} proves that for the majority of integrators used in practice only generally symplectic time integration techniques are able to conserve quadratic invariants exactly and that any method that conserves an approximate invariant exactly must be conjugate to a symplectic method.
Our results verify the theory and also indicate that certain explicit methods that approximate symplectic methods also have relatively good conservation properties, perhaps because they require small times steps for stability.

However, an extensive set of numerical tests demonstrate that the schemes that conserve both energy and enstrophy exactly lead to turbulent simulations with unphysically high amplitudes at short spatial wavelengths.
Reintroducing dissipation in a  number of different ways, e.g. by adding numerical viscosity or by using an asymmetric implicit time integration method, like backward Euler, tends to tame the short wavelength part of the spectrum.
Unfortunately, controlling the spectrum through a non-symplectic time-integrator spoils conservation in manner similar to adding an artificially large viscosity.
For discontinuous Galerkin methods there is another choice: upwinding the advection operator within the FEEC framework leads to a numerical scheme that dissipates enstrophy but conserves energy. This combination appears to offer the best compromise of capturing the correct mesoscopic spatial mode structure by allowing mid-range scales to grow to the proper amplitude, while still  controlling the spectrum of the short wavelength spatial modes. 
The~upwinded DG can be seen as the stream function/vorticity 
formulation of the Lie derivative method \cite{NataleIMA2018} and also provides 
a simpler alternative to the SUPG formulation of stream function/vorticity introduced in \cite{NataleQJRMS2017}. We also observe, that the~implicit symplectic time integrators match very precisely the~results by explicit integrators, while acquiring $\sim 10 \times$ speed-up. This means that the CFL condition is too restrictive in order to capture the~relevant physics in mesoscopic spatial mode range.

In order to verify and explore the theoretical results, a new FEEC fluid solver based on the MFEM finite element library \cite{MFEM}, a highly scalable FEEC implementation,  has been developed.
The solver is written in C++ and has been run within a high-performance computing environment on thousands of cores.
Both linear and nonlinear solvers for the equations have been developed, including preconditioning strategies for the linear solves.

A series of numerical test problems that display both decaying and forced 2D turbulence, recently reviewed in \cite{Boffetta2007jfm, Boffetta2010pre, Boffetta2012arfm}, are studied to examine the properties of the solution.
In order to do so, we reproduce the direct and inverse cascades predicted by Batchelor and Kraichnan \cite{Batchelor1959jfm, Kraichnan1967pf, Kraichnan1971jfm, Kraichnan1980rpp}.
The numerical model is verified by comparing the linear growth rates with predictions of a semi-analytical eigensolver for various initial and boundary conditions.
The nonlinear results are benchmarked by comparing the Fourier spectrum to the results of the Global Drift Ballooning (GDB) finite difference code \cite{Zhu2018cpc}, which has been extensively verified for plasma physics problems.

In the next section, the model equations for the 2D incompressible Navier-Stokes equations and the 2Dd rift-reduced MHD equations, which reduce to the Navier-Stokes equations in the appropriate limits are introduced.
The conservative FEEC schemes are derived  in  \Cref{sec:method}, beginning with a review of the need for symplectic time integration techniques.
 \Cref{sec:results} presents  our numerical results.
The linear and nonlinear behavior of the various  schemes are verified in  \Cref{sec:kh_instability} for a series of simulations of decaying Navier-Stokes turbulence.
These results explicitly demonstrate that the choice of numerical scheme has an important impact on the power spectrum of the mid to high spatial mode numbers and that the~symplectic time integrators show $\sim 10\times$ speed-up compared to standard explicit integration.
Forced 2D Navier-Stokes turbulence simulations are presented in  \Cref{sec:forced_turbulence} that verify the well-known predictions of the direct and inverse cascades of energy and enstrophy.
Again, it will be shown that the choice of numerical methods has an  impact on the short wavelength range of the power spectrum.
Our conclusions are summarized in the final section.

\section{Model equations \label{sec:model}}

Conservation of mass is expressed as the advection of the mass density $\rho$ by the fluid velocity $\vDrift$ 
\begin{align}     
  \partial_t \rho +  \Vnabla\cdot  \rho  \vDrift &= 0
  .
\end{align}
Because the time-rate of change of the mass density along the fluid trajectory is given by
\begin{align*}
  d\rho/dt :=  \partial_t \rho  + \vDrift\cdot \Vnabla  \rho   &= -\rho\Vnabla \cdot \vDrift,
\end{align*}
the condition that the  fluid flow be incompressible requires  the velocity to be divergence free $\Vnabla\cdot \vDrift=0$.
Momentum conservation is expressed via the   incompressible Navier-Stokes equation 
\begin{align} 
  \partial_t (\rho \vDrift) + \Vnabla\cdot  (\rho  \vDrift  \vDrift) + \Vnabla p&= 
     \Vnabla \cdot \mat{\mu}  \cdot \Vnabla  \vDrift  
    - \alpha \rho \vDrift +\VF_v  
\end{align}
where $p$ is the scalar pressure,  $\mat{\mu}$ is the viscosity tensor, 
$\alpha$ is a frictional damping rate, and $\VF_v$ is a ``stirring force".
Enforcing the constraint of incompressible flow requires the pressure to satisfy the relation
\begin{align} 
  \Vnabla\cdot \rho^{-1} \Vnabla p&=  \Vnabla \cdot \rho^{-1} \left(\Vnabla \cdot \mat{\mu}  \cdot \Vnabla  \vDrift   
    - \alpha \rho \vDrift + \VF_v   -   \rho  \vDrift \cdot \Vnabla \vDrift \right)
    \label{eq:NS_momentum}
    .
\end{align}

In 2D, one can ensure that the velocity is exactly divergence free by assuming that it is defined via
\begin{align} 
\vDrift=\zhat \times \Vnabla \Phi
\end{align} 
where $\Phi$ is the stream function
and solving for the vorticity
\begin{align} 
\fvor = \Vnabla \cdot (\zhat \times \vDrift) =- \Vnabla \cdot (\mat{1} - \zhat \otimes \zhat) \cdot \Vnabla  \Phi.
\end{align} 
For constant density, $\rho$,  and constant transport coefficients, $\alpha$ and $\mu$, the potential formulation allows one to rewrite the equations as 
\begin{align} 
  \pdv{ \fvor  }{t} + \Vnabla\cdot (\fvor \vDrift) &= 
     \Vnabla \cdot \mu \Vnabla  \fvor  - \alpha \fvor   + S_{\fvor}  
     \label{eq:NS_vorticity}
\end{align}
where $ S_{\fvor}:=-\zhat \cdot \Vnabla \times \VF_v/\rho $.

The  drift-reduced MHD model is derived using an adiabatic approximation to the drift velocity that assumes that the spatial and temporal variations of the magnetic field, $\VB$, are much slower than the dynamics of interest.
For the models studied here  \cite{Hasegawa1978Mima, Horton1994}, the electric field $\VE$, will be assumed to be approximately electrostatic,  i.e. determined via $\VE=-\Vnabla \phi$ where $\phi$ is the electric potential.  
Generalizations of this model \cite{Drake1984Antonsen, Zeiler1996, Hazeltine2003Meiss, Xu2000, Simakov2003, Simakov2004} can also be used to describe the physics of shear Alfv\'en waves and magnetic reconnection.
The dominant momentum balance is ensured by requiring the Lorentz force, proportional to $\VE+\Vv\times \VB$, to vanish to lowest order.
This implies that the drift velocity is approximately given by the $\VE\times\VB$ velocity
\begin{align}
\vDrift =\VE\times\frac{\bhat}{B} =\frac{\bhat}{B} \times \Vnabla \phi
.
\end{align} 
where  $\bhat=\VB/B$ is the direction of the magnetic field and $\Bmag=\norm{\VB}$ is its magnitude.
Although the drift-reduced MHD velocity is not exactly divergence free for spatially varying $\Bmag$, it is approximately divergence-free in the sense that the scale length for variations is much longer than that of turbulent fluctuations. 
The model is exactly divergence-free when $\VB$ is constant in space, e.g. $\VB=\Bmag_0\zhat$ for constant $\Bmag_0$, in which case it reduces to the same form as for the incompressible neutral fluid.

Assuming that the plasma dynamics is quasineutral requires the electric current, $\VJ$, to be divergence free: $\Vnabla\cdot\VJ=0$. 
Solving for the velocity to one higher order yields a polarization current proportional to the rate of change of the drift velocity
\begin{align}
\vDrift =\VE\times\frac{\bhat}{B} +\frac{\bhat}{B} \times \left[\partial_t +\Vnabla\cdot \Vv \right] \left(\rho \VE\times\frac{\bhat}{B}\right) + \dots
.
\end{align} 
Using the quasineutrality constraint $\Vnabla\cdot \VJ_\|=\Vnabla\cdot (\bhat \sigma_\| E_\|) = -\Vnabla\cdot\VJ_\perp $ yields the dynamics of the charge density $w=\Vnabla\cdot \frac{\bhat}{B} \times \rho \vDrift$.
Truncating the expansion at the first nontrivial order yields a closed set of equations
\begin{subequations}\label{eq:model}
\begin{align}
  \pdv{\fvor}{t} +\Vnabla \cdot  \fvor \vDrift &= 
    \Vnabla \cdot \mat{\mu} \cdot \Vnabla \fvor - \alpha \fvor + \Vnabla_\| \cdot \sigma_\| \Vnabla_\|\phi + S_{\fvor}
  , \label{eq:vorticity}   
  \\
  \Vnabla \cdot \frac{\rho}{\Bmag^2} \VnablaPerp \fpot 
    &= - \fvor 
  , \label{eq:potential}
  \\
  \vDrift &= \frac{\bhat}{\Bmag} \times \VnablaPerp \fpot
   \label{eq:driftV}
\end{align}
\end{subequations}
Here $\rho$ is the mass density, $\mat{\mu}$ is the viscosity tensor, 
 $\alpha$ represents a frictional damping coefficient, e.g. friction with a neutral gas assumed to be at rest, and $\sigma_\|$ represents electrical conductivity parallel to magnetic field lines.
Here, directionality with respect to the~magnetic field 
$\VB$ is expressed using the conventions 
\begin{align*}
\VnablaPara &:= \bhat\otimes\bhat \cdot \Vnabla \\
 \VnablaPerp &:=  \Vnabla -\VnablaPara  = \left( \mat{1} - \bhat\otimes\bhat \right) \cdot \Vnabla 
\end{align*} 
where $\mat{1}$ is the identity matrix.
Clearly, these equations must be supplemented with an appropriate set of initial conditions and boundary conditions in order for the solution to be unique.

  For the incompressible model, kinetic energy, $\Energy$, which we simply refer to as energy,  is defined as
  \begin{gather}
   \Energy :=  \frac{1}{2} \int_\Omega \rho \vDrift^2 ~dV  =  \frac{1}{2} \int_\Omega \frac{\rho}{\Bmag^2} \VnablaPerp \fpot \cdot \VnablaPerp \fpot ~dV
  \end{gather}
  where $d\Vol$ is the 2D volume element, i.e. the area, for the region of interest $\Omega$. 
  Similarly, enstrophy, $\Enstrophy$, is defined as
  \begin{gather}
   \Enstrophy :=  \frac{1}{2} \int_\Omega \fvor^2 ~dV 
   .
  \end{gather}
Both energy and enstrophy are conserved in  the absence of dissipation and external forces when  an appropriate set of boundary conditions is chosen that eliminates the net fluxes over the boundary of the region of interest, $\partial \Omega$.
  
For both of the vorticity equations \eqref{eq:NS_vorticity} and \eqref{eq:vorticity}, aside from the forcing terms, the right-hand side represents a set of terms that dissipate both energy and enstrophy. 
This is easily proven by noting that, for an appropriate set of boundary conditions, the dissipative terms on the right hand side are negative semi-definite operators.
Hence, the conservative part of the equation is simply the left hand side, the advection of the vorticity by the flow
\begin{align}
  \pdv{\fvor}{t} +\Vnabla \cdot \fvor \vDrift &= \pdv{\fvor}{t} +\vDrift \cdot \Vnabla   \fvor =0
   \label{eq:vorticity_conservative}   
   .
\end{align}
It is precisely the conservative discretization of the advection operator that will be investigated theoretically in Sec. \ref{sec:space}.

\section{Conservative Methods for Arbitrary Polynomial Order \label{sec:method}}

\subsection{Overview of Results}
Conservation of energy and enstrophy can each be posed as the preservation of a quadratic invariant in time, also known as a quadratic first integral.
Thus, it is critical to consider conservation of quantities of interest in both the discrete space and discrete time setting.
In this work, we explore numerical approaches that use the method of lines, i.e. the tensor product of a spatial finite element discretization and a temporal discretization based on a  numerical time integration method.
This allows one to treat the spatially discretized system as a large set of ordinary differential equations (ODEs) and to choose the approaches to spatial and temporal discretization independently.

In the next section, we treat the temporal discretization by recalling the theorems of Chartier, Faou, and Murua (2006) \cite{Chartier.2006} that prove that, for the majority of numerical integration methods used in practice,  in order to exactly preserve a quadratic invariant, the method must either be a generally symplectic method or must be conjugate to a generally symplectic method (see \cite{Hairer2006book} for a comprehensive overview).
The general symplecticity conditions \cite{Calvo1994}  require a symplectic method to  be fully implicit Runge-Kutta Gauss method.
The simplest methods that preserve quadratic invariants are the second order implicit midpoint method, which is symplectic, and the implicit trapezoidal method, which is conjugate to implicit midpoint.
We also discuss the ability of certain explicit and multistep schemes to approximately conserve invariants (\Cref{sec:explicit}), essentially because they 
approximate symplectic methods such as implicit midpoint.

Then,  in \Cref{sec:space} we treat the spatial discretization by  proving that, in the continuous time setting, spatial
discretization with the appropriate choice of finite-element
spaces leads to an energy and enstrophy conserving method.
This leads to the result that using a symplectic time integration method for the resulting set of discretized ODEs yields the desired conservation laws.
Finally,  for reference, Sec. \ref{sec:poisson_bracket} explicitly expresses these results in terms of the Poisson bracket, a form for the advection operator that is commonly used in the fluid dynamics and  plasma physics literature.

\subsection{Exactly Conservative Time Integrators: Fully Implicit Symplectic Methods \label{sec:temporal} }
For time-dependent simulations, it is critical to consider conservation of
quantities of interest in both the discrete space and time setting.
In Sec. \ref{sec:space}, we prove that the finite element spatial discretization leads to a discrete system of ODEs that possess the quadratic invariants of  energy and enstrophy.
Consider the general system of ODEs
\begin{equation*}
  \frac{d  \vect{u}}{dt}   = \vect{f}(\vect{u})
  .
\end{equation*}
A conserved quantity, $C( \vect{u},t)$, is constant in time, i.e.  invariant under  translations in time, if it satisfies $dC/dt=0$.
 A quadratic form is defined by a
symmetric matrix $\matr{Q}$ via
\begin{equation*}
C =\tfrac{1}{2} \vect{u}^T\matr{Q}\vect{u}
\end{equation*}
and it defines a quadratic invariant if it satisfies
\begin{equation} \label{eq:quad}
\frac{d C}{dt}  =  \vect{u}^T\matr{Q} \frac{d  \vect{u}}{dt} = \vect{u}^T\matr{Q}\vect{f}(\vect{u}) = 0,
\end{equation}
for all solutions $\vect{u}$ in the domain of $\vect{f}$ (see, e.g.,
\cite{Cooper.1987}); i.e. for all solutions $\vect{u}(t)$ to the set of ODEs.
We will show that with careful construction of the FEM spaces and discretization,
energy and enstrophy conservation can both be expressed in the form
of \eqref{eq:quad}.

The monograph of Hairer, Lubich, and Wanner \cite{Hairer2006book} provides an excellent review of numerical integrators that conserve nonlinear invariants as well as of numerical integrators in general.
In fact, the construction of time integrators that conserve nonlinear invariants is one of the primary motivations for the development of geometric numerical integration.
 In the following we highlight a few of the main points that are important for our conclusions.
 
The pioneering work of Butcher \cite{Butcher1972} and Hairer and Wanner \cite{Hairer1974} enabled the analysis
of many widely used time integration techniques, including Runge-Kutta (multi-stage) methods and multi-step methods, and led to the introduction of general linear methods.
Like the exact solution, the result of these methods can be expressed as a ``B-series'' defined as a formal power series over the set of elementary differentials.
Because the elementary differentials are  indexed by the graphs of ``rooted trees,'' graph-theoretical methods can be brought to bear on the analysis \cite{Hairer2006book}.
The results can also be extended to partitioned-series or ``P-series'' integration methods \cite{Hairer2006book}  that partition the system into two sets of variables, $\vect{u} = \left[\vect{q},\vect{p}\right]^T$, where each part uses a different numerical integrator.
In this case, the elementary differentials can be indexed by trees with roots of two different colors.
Moreover,  B-series and P-series form a group under composition, so the results can be extended to any composition of methods.
Thus, compositions of B-series and P-series methods contain a large portion of the numerical time integration techniques used in practice.
 
For all B-series time integrators, quadratic invariants are only
conserved   if the integration method is a generally symplectic  method  \cite[Theorem 1]{Chartier.2006}:
 \begin{theorem}[Only symplectic integrators conserve quadratic invariants]
   \label{th:symplectic_integrators}
   Any B-series time integrator that exactly conserves a quadratic invariant must be a generally symplectic integrator.
  \end{theorem}

The general symplecticity conditions for B-series integrators \cite{Calvo1994}  require the method to  be a symplectic Runge-Kutta method (originally called  ``orbitally stable'' Runge-Kutta methods by \cite{Cooper.1987}).
Symplectic Runge-Kutta methods must be of the Gauss-Legendre (Gauss) type  \cite{Lasagni1988, SanzSerna1988, Suris1988}.
The  Gauss-Radau and Gauss-Lobatto Runge-Kutta methods fail to be symplectic due to their lack of symmetry in time.  
Thus, the simplest symplectic method is the second-order implicit midpoint method.
Note, however, that no Runge-Kutta method can conserve all polynomial invariants for polynomial order $>2$ \cite{Hairer2006book}.

It is clear that a spatial discretization will only preserve the conservation of a quadratic invariant under certain conditions, e.g. to be derived in \Cref{sec:space}. This immediately implies
  \begin{corollary}  \label{cor:symplectic_integrators}
  For B-series time integrators, only generally symplectic methods coupled with a
  conservative spatial discretization are guaranteed to conserve
  enstrophy and energy exactly.  
  \end{corollary}

Integration methods that are conjugate to symplectic method have the form $ \Phi_h = \chi_h^{-1} \circ \Psi_h \circ \chi_h$, where $\Psi_h(\vect{u})$ is a symplectic transformation for time step $h$, and $\chi_h(\vect{u})$ is a near identity transformation of order $\CO(h)$  in a neighborhood of the initial condition.  This form ensures that that application of multiple steps still leads to a result that is $\CO(h)$ close to the symplectic result, i.e. $\Phi_{n h}  = (\Phi_h\circ)^n= \chi_h^{-1} \circ \Psi_{nh} \circ \chi_h$.  Clearly, if an invariant  is  exactly conserved by the symplectic integrator, $\Psi_h\circ C = C $, then there is an approximate invariant, $C_h=\chi_h^{-1} \circ C$, that is exactly conserved, $\Phi_{h}\circ C_h = C_h$.
In fact,  conjugate symplectic  methods are the only methods have this property \cite[Theorem 4]{Chartier.2006}:
 \begin{theorem}[Only conjugate symplectic methods conserve approximate quadratic invariants]
   \label{th:approximate_integrators}
    Any B-series time  integrator that exactly conserves an approximate quadratic  invariant must be conjugate to a generally symplectic integrator.
 \end{theorem}
The simplest and most practical conjugate symplectic method is the second order implicit trapezoidal integrator, a Gauss-Lobatto method that is conjugate to implicit midpoint.

The same line of reasoning regarding spatial discretization leads to:
  \begin{corollary}  \label{cor:approximate_integrators}
  For B-series time integrators, only generally symplectic methods coupled with 
  a conservative spatial discretization are guaranteed to exactly conserve 
  invariants that represent enstrophy and energy. 
  \end{corollary}

Chartier, Faou, and Murua \cite{Chartier.2006} also extend these two theorems to show that conservative P-series methods must either be symplectic or conjugate-symplectic.
However, proper P-series methods  can only conserve quadratic forms that, like angular momentum, are linear in each of the partitioned variables, i.e. the invariant must be of the form of, $\vect{q} \matr{E}\vect{p}$ with constant matrix $\matr{E}$.
For example, Gauss-Lobatto IIIA-IIIB pairs conserve all quadratic invariants of this form \cite{Sun1993}.
However, any P-series method that conserves general quadratic invariants such as energy and enstrophy must actually be a B-series method \cite{Chartier.2006}, and, hence, must be a Gauss-Legendre Runge-Kutta method.

Note that these methods must be fully implicit as opposed to  diagonally implicit Runge-Kutta (DIRK) methods.
In order for an implicit method to provide exact conservation,
an exact solve of each of the stage equations is required. 
Inexact solution methods such as iterative solvers can potentially degrade perfect conservation.
For example, one must take care to require sufficiently accurate
linear solves within each nonlinear stage to maintain good conservation properties.

Moreover, although Gauss and other fully implicit Runge-Kutta methods
can offer  high accuracy in addition to good conservation properties,
historically, high-order methods have rarely been used in practice
for numerical PDEs due to the difficulty of solving the fully coupled
stage equations. However, recent theoretical and algorithmic
developments have made such integration tractable and even faster
than diagonally implicit Runge-Kutta (DIRK) methods at times \cite{irk1,irk2}.

Finally, we note that, since the finite element method presents the equations in variational form, it should be of interest to explore variational integrators.
Variational integrators are typically partitioned partially implicit methods, but 
in special cases, e.g. when the Hamiltonian is separable in kinetic and potential energy, variational methods can even be explicit.
Since the 2D Euler equations are a Lie-Poisson Hamiltonian system, with a degenerate Poisson bracket, specialized variational integrators must be used \cite{Ge1988, Channell1991,Marsden2001}.

\subsection{ Approximately Conservative Time Integrators: Explicit \& Multistep Methods \label{sec:explicit}}

As discussed above, Runge-Kutta schemes must be fully implicit in order
to conserve quadratic invariants.
Our numerical results in \Cref{sec:results} demonstrate
that, with effective solvers, implicit methods can notably outperform explicit
methods in terms of wall-clock time as well. However, in some instances, such as
problems with less restrictive time step constraints or software packages that
want to avoid the machinery of implicit integration, explicit integration may
still be desirable.

Numerical results show that explicit midpoint and other methods that approximate the  
implicit midpoint method can still have relatively good approximate
conservation properties for quadratic invariants.
This may be the case because explicit integration tends to require very small time
steps for stable integration, and so these approximations may have relatively small error for each time step. 
Due to the fact that the implicit trapezoidal method is conjugate to implicit midpoint, one would expect similar behavior for explicit methods that approximate implicit trapezoidal.
As yet another example, the GDB code \cite{Zhu2018cpc} uses the explicit  ``trapezoidal leapfrog method'' \cite{Zalesak1979jcp, Jardin2010book}, a second order accurate general linear method that stabilizes the leapfrog method by using an explicit trapezoidal predictor step to improve the estimate for the value at the midpoint.  
In this case, the explicit trapezoidal predictor step approximates implicit trapezoidal, while the final result approximates implicit midpoint.
Our numerical results in \Cref{sec:results} demonstrate that both the explicit midpoint method and the trapezoidal leapfrog  method have relatively good conservation properties.

A number of additional time-reversible multistep methods were identified by Ref. \cite{Evans.1999} that have good conservation properties, even though no multistep method can be generally symplectic \cite{Hairer1998}.
The two simplest methods are (i) the implicit trapezoidal method, which is conjugate to implicit midpoint, and (ii) the leap-frog method (which they referred to as explicit midpoint).   
 For  the system of ODEs $d\vect{u}/dt
= \vect{f}(\vect{u})$, the leapfrog multistep method \cite{SanzSerna1985} is simply
\begin{align*}
  \vect{u}_{n+1} &= \vect{u}_{n-1} +  \vect{f}(\vect{u}_{n }) 2 \delta t 
  .
\end{align*}
If the system of ODEs preserves the quadratic invariant $C= \tfrac{1}{2} \vect{u}^T(t)  \matr{Q}  \vect{u}(t)$, then the leapfrog method preserves the invariant 
$
  C = \tfrac{1}{2} \vect{u}_{n}^T  \matr{Q}  \vect{u}_{n+1}  .
$
This is easy to show by using the symmetry of $\matr{Q}$ to find \cite{SanzSerna1985}
\begin{equation*}  
  \vect{u}_{n}^T  \matr{Q} \vect{u}_{n+1} - \vect{u}_{n}^T  \matr{Q}  \vect{u}_{n-1} =  \vect{u}_{n}^T   \matr{Q}    \vect{f}(\vect{u}_{n})  \delta t=0
  .
\end{equation*}
The good conservation properties result from the fact that the underlying one-step method is symmetric in time, and, hence, by Theorem 1.1 of Ref. \cite{Hairer2008jcm}, this one-step method is conjugate to a symplectic method.
  
It should be pointed out that the leapfrog method is known
to have nonlinear instabilities, as analyzed in \cite{SanzSerna1985,expmidpoint2}
(observed as early as 1965 \cite{lilly1965computational}). 
In fact, explicit general linear methods all have similar nonlinear
instabilities, analogous to the leapfrog method \cite{Butcher.2016}.
They also generically suffer from linear instabilities, known as parasitic modes, unless, like leapfrog, the methods satisfy the  ``zero-growth'' conditions \cite{Evans.1999}.
 In fact, Ref. \cite{Evans.1999} identified all of the time-reversible multi-step time-integration of sixth order and below that satisfy the ``zero-growth'' conditions for parasitic modes.
For problems or software packages where explicit integration is desirable and nonlinear instabilities can be avoided, these zero-growth methods may have the
highest order accuracy and the best conservation properties.
Again, the simplest of these is the second-order leapfrog method and one could also consider adaptive or modified versions
of leapfrog as discussed in \cite{sanz1982explicit}.

\subsection{ Spatial Discretization: Arbitrary Order Finite Element Methods \label{sec:space}}

In this section, we derive conservative finite element discretizations for Eqs.
\eqref{eq:model}.
In the next two subsections, we prove that the spatially discretized equations are conservative if the potential and vorticity are chosen from the appropriate finite element spaces.
This leads to the result that the  combination of correctly chosen spatial and temporal discretization is conservative.

Recall that a finite element discretization presents a  set of partial differential equations in integral form.
First let us establish some notation: let $\Omega$ be the~computational domain, $\partial\Omega$ be the boundary of the domain,
and $\vec{n}_{\partial\Omega}$ be the outward normal to the~boundary.
Let $\tauh$ be a~mesh covering $\Omega$, specified by finite element cells $K$, and their boundaries $\partial K$, as well as nodal interpolation points within each element.
Let $\Gamma$ denote the~union of the~internal faces between elements of the~mesh.
The mesh $\tauh$ also defines a uniquely oriented unit vector, $\vec{n}_\Gamma$,  normal to every face of $\Gamma$.
A weak formulation allows the solutions and their partial derivatives to manifest discontinuities across $\Gamma$.
For the weak integral formulation of the partial differential equations, the effects of the discontinuities are carefully handled using various numerical strategies.

\begin{definition}[Jump and average operators] The~jump and average operators are defined on the element interfaces $\Gamma$ via
\begin{align} \label{eq:jump_average_ops}
  \jumpOp{ \ftestLtwo } & =\ftestLtwo^+ - \ftestLtwo^-
  & \avgOp{ \ftestLtwo } &=  ( \ftestLtwo^+ + \ftestLtwo^-) /2
  & \mathrm{on}~\Gamma
\end{align}
where $\ftestLtwo^+$ and $\ftestLtwo^-$ are upwind and downwind values of
$\ftestLtwo \in \Ltwo{}$ with respect to $\vec{n}_\Gamma$.
\end{definition}

Finite element exterior calculus crucially relies on the
the~jump operator defined in Eq. \eqref{eq:jump_average_ops}, which
measures the lack of {\it continuity}, the essential property defining the spaces of discrete differential forms.
The space of discrete differential zero-forms, $\Hone{}$, corresponds to the~finite element space
of continuous piecewise polynomial scalar functions, defined by $\jumpOp{\ftestHone} = 0$ for
any function $\ftestHone \in \Hone{}$.
The notation $\Hone{\ppot}$ denotes continuous piecewise  polynomial functions of order $\leq \ppot$.
In contrast, the space of  discrete differential volume-forms (three-forms), $\Ltwo{}$, corresponds to the~finite element space
of piecewise polynomial functions that are potentially discontinuous  $\jumpOp{\ftestLtwo}\neq  0$ across $\Gamma$, for
any function $\ftestLtwo \in \Ltwo{}$.
The notation $\Ltwo{\ppot}$ denotes piecewise polynomial functions of order $\leq \ppot$ that are potentially discontinuous across $\Gamma$.

The space of piecewise polynomial vector fields that are potentially discontinuous across $\Gamma$ is denoted $\LtwoVec{}$ for dimension $d=3$.
On the surface $\Gamma$, a vector field,
$\vec{v} = \vec{v}_{n} + \vec{v}_{t}$, can be decomposed into components
normal,  $\vec{v}_{n} = (\vec{v} \cdot \vec{n}_{\Gamma}) \vec{n}_{\Gamma} $, and
tangential, $\vec{v}_{t} = \vec{v} - \vec{v}_{n}$, to  $\Gamma$.
For a Cartesian mesh, the space of discrete differential one-forms corresponds to the~finite element space,
$\Hcurl{}$, of piecewise polynomial vector functions with a continuous tangential component,
$\jumpOp{\vec{v}_t} = 0$, but a potentially discontinuous normal component across $\Gamma$.
Similarly, the space of discrete differential two-forms corresponds to the~finite element space,
$\Hdiv{}$, of piecewise polynomial vector functions with a continuous normal component,
$\jumpOp{\vec{v}_n} = 0$, but a potentially discontinuous tangential component across $\Gamma$.
In \cite{Ndlec1980MixedFE,RiebenThesis2004,CastilloCMES2004} it
was shown how arbitrary order differential forms can be constructed on general
non-Cartesian meshes. In this case,
the~conditions on continuity require that the integral of the jump over each
face $e$ of $\Gamma$ must vanish, i.e.
$\int_{e} \jumpOp{\vec{v}} \cdot \vec{n}_\Gamma q dS = 0$ for any two-form
$\vec{v} \in \Hdiv{}$ and any zero-form $q \in \Hone{}$, and
$\int_{e} \jumpOp{\vec{v}} \cdot (\vec{n}_\Gamma \times \vec{q}) dS = 0$
for any two one-forms $\vec{v}, \vec{q} \in \Hcurl{}$. Note that $q$ and
$\vec{n}_\Gamma \times \vec{q}$ are single valued on $e$.
Clearly, the space of discontinuous vector fields is the union of vector spaces $\LtwoVec{}~=~\Hcurl{}~\cup~\Hdiv{}$.
This is the discrete Hodge decomposition theorem (also known as the Helmholtz decomposition in 3D).

For the problems of interest, where the flows are nearly incompressible, it will be assumed that $\vDrift \in \Hdiv{}$, so that the normal component $ \vDrift \cdot \vec{n}_\Gamma$ across $\Gamma$  must be continuous.
\begin{lemma}[Divergence of the drift velocity]
  \label{thm:vDriftIncompressible}
  If $\fpot \in \Hone{}$ and $ \bhat/\Bmag \in \Hcurl{}$,
  then the~discrete approximation of the drift velocity \eqref{eq:driftV} implies that $\vDrift\in\Hdiv{}$.
  Furthermore, when $\Vnabla \times \frac{\bhat}{\Bmag} = \vec{0}$,
  $\vDrift$ is divergence-free at the discrete level: $\Vnabla \cdot \vDrift = 0$.

\end{lemma}
\begin{proof}
  The FEEC relations imply that the usual theorems of vector calculus hold at the discrete level.
  The assumption $\fpot \in \Hone{}$ implies  $\Vnabla \fpot \in \Hcurl{}$.
  If $\vec{a}, \vec{c} \in \Hcurl{}$, then $\vec{a} \times \vec{c} \in \Hdiv{}$.
  Hence, the assumptions imply that $\frac{\bhat}{\Bmag} \times \Vnabla \fpot \in \Hdiv{}$.
  If $\vec{c} \in \Hcurl{}$, then $\Vnabla \cdot \Vnabla \times \vec{c} = 0$.
  Thus, the assumptions $\frac{\bhat}{\Bmag} \in \Hcurl{}$ and $\fpot \in \Hone{}$, imply that
  $\frac{\bhat}{\Bmag} \fpot \in \Hcurl{}$.
  To prove the final relation, note that if $\Vnabla \times \frac{\bhat}{\Bmag} = \vec{0}$,
  then $\vDrift = \Vnabla \times \left( \frac{\bhat}{\Bmag} \fpot \right)$ and, thus,
  $\Vnabla \cdot \vDrift = 0$.
\end{proof}
Consistency with the condition 
$\vDrift\in \Hdiv{}$ requires $\Vnabla \phi \in \Hcurl{}$, hence,
the potential to be continuous $\fpot \in \Hone{}$.
The~weak formulation of the~potential equation \eqref{eq:potential}
in the~finite element space of piecewise continuous polynomial functions of 
order $\ppot$, $\fpot \in \Honep$, yields
\begin{equation}
  \int_\Omega \frac{\rho}{\Bmag^2} \VnablaPerp \ftestHone
    \cdot \VnablaPerp \fpot ~dV
  =
  \int_\Omega \ftestHone \fvor  ~dV
  +
  \int_{\partial\Omega} \ftestHone \frac{\rho}{\Bmag^2}
    \VnablaPerp \fpot \cdot \vec{n}_{\partial \Omega} ~dS
  ,
  \hfill~\forall \ftestHone \in \Honep
  \label{eq:potentialWeak}
  .
\end{equation}
Similarly, the weak continuous formulation of polynomial
order $\qvor$ for the vorticity \eqref{eq:vorticity}, $\fvor \in \Honeq$, yields
\begin{multline}
  \int_\Omega \ftestHone \left( \partial_t \fvor  + \vDrift \cdot \Vnabla \fvor\right) ~dV   =
  \\
    \int_\Omega \left( \ftestHone S_{\fvor} - \ftestHone \alpha \fvor
    -  \Vnabla \ftestHone \cdot \mat{\mu} \cdot \Vnabla \fvor - \Vnabla \ftestHone \cdot \sigma_\| \Vnabla_\|\fpot \right) ~dV
    \\
      +  \int_{\partial\Omega} \ftestHone  \left(  \mat{\mu} \cdot \Vnabla \fvor +  \sigma_\| \Vnabla_\| \fpot  \right)  \cdot \vec{n}_{\partial \Omega} ~dS
  ,
   \hspace{0.25in}  ~\forall \ftestHone \in \Honeq
   \label{eq:vorticityH1}
   .
\end{multline}

Next, we derive the weak discontinuous Galerkin (DG)   finite element formulation of the vorticity.
In order to do so, let us remind the reader of a few important facts. 
First, the jump of a product of functions  $a, b \in \Ltwo{}$ is
  \begin{align}\label{eq:productJump}
    \jumpOp{ a b  } = \avgOp{ a } \jumpOp{ b }
    + \jumpOp{ a }  \avgOp{ b } .
  \end{align}
  This result follows upon substitution of the definitions \eqref{eq:jump_average_ops}.
Second, the divergence theorem for 
  $a  \in \Ltwo{}$  and $\vDrift \in \Hdiv{}$ yields
  \begin{gather} \label{eq:divTheoremDG}
    \int_\Omega \Vnabla\cdot(a   \vDrift) ~dV
    = \int_{\partial\Omega} a   \vDrift \cdot \vec{n}_{\partial \Omega} ~dS
    + \int_\Gamma  \jumpOp{ a  } \vDrift \cdot  \vec{n}_\Gamma~dS .
  \end{gather}
 This can be shown by using the divergence theorem on the finite element mesh $\tau_h$,
  \begin{gather*}
    \int_\Omega \Vnabla \cdot (a  \vDrift) ~dV
    = \sum_{K\in \tau_h} \int_{\partial K}   \Vnabla \cdot (a   \vDrift) ~dV
    = \sum_{K\in \tau_h} \int_{\partial K} a  \vDrift \cdot  \vec{n}_{\partial K} ~dS.
  \end{gather*}
  When $\vDrift \in \Hdiv{}$, $\jumpOp{ \vDrift } \cdot  \vec{n}_\Gamma= 0$, hence
  $\jumpOp{ a   \vDrift } \cdot  \vec{n}_\Gamma=  \jumpOp{ a   } \vDrift \cdot  \vec{n}_\Gamma$,
  and the~sum of discontinuities can be expressed using the jump operator
  \begin{gather*}
    \sum_{K\in \tau_h} \int_{\partial K} a   \vDrift \cdot \vec{n}_{\partial K} ~dS
  =
    \int_{\Gamma}  \jumpOp{ a  }  \vDrift \cdot \vec{n}_\Gamma~dS
    + \int_{\partial \Omega} a   \vDrift \cdot \vec{n}_{\partial \Omega} ~dS .
  \end{gather*}
  Combining the two equations above yields \eqref{eq:divTheoremDG}. 
  Finally, combining Eq.~\eqref{eq:productJump} with Eq.~\eqref{eq:divTheoremDG} leads to
  \begin{gather} \label{eq:divTheoremDGproduct}
    \int_\Omega \Vnabla\cdot(a b  \vDrift) ~dV
    = \int_{\partial\Omega} a b   \vDrift \cdot \vec{n}_{\partial \Omega} ~dS
    + \int_\Gamma \left(\jumpOp{ a  } \avgOp{b} + \avgOp{a} \jumpOp{b} \right) \vDrift \cdot \vec{n}_\Gamma~dS .
  \end{gather}

The weak discontinuous formulation for the vorticity, $\fvor \in \Ltwoq$, allows one to upwind the side on which the flux is sampled relative to the faces of $\Gamma$.
Thus, the weak DG formulation of polynomial order $\qvor$ for the vorticity equation \eqref{eq:vorticity} yields additional jump terms   that define the numerical fluxes
\begin{multline}
  \int_\Omega \ftestLtwo  \left( \partial_t \fvor  + \vDrift \cdot \Vnabla \fvor\right) ~dV  
    -  \int_{\Gamma} \left[ \vDrift \cdot \vec{n}_\Gamma \avgOp{ \ftestLtwo }  \jumpOp{ \fvor }
    - \delta \abs{\vDrift\cdot \vec{n}_\Gamma} \jumpOp{ \ftestLtwo }  \jumpOp{ \fvor } \right]~dS =
  \\
     \int_\Omega \left( \ftestLtwo S_{\fvor} - \ftestLtwo \alpha \fvor
    - \Vnabla \ftestLtwo \cdot \mat{\mu} \cdot \Vnabla \fvor
    - \Vnabla \ftestLtwo \cdot  \sigma_\| \Vnabla_\|\fpot \right) ~dV
    \\
    +  \int_{\partial\Omega} \ftestLtwo  \left( \mat{\mu} \cdot \Vnabla \fvor 
      +   \sigma_\| \Vnabla_\| \fpot  \right)  \cdot \vec{n}_{\partial \Omega} ~dS
\\
+ \int_\Gamma \left[
    \jumpOp{ \ftestLtwo }
    \vec{n}_\Gamma \cdot \mat{\mu} \cdot \avgOp{ \Vnabla \fvor }
    + \avgOp{ \Vnabla \ftestLtwo }
    \cdot \mat{\mu} \cdot \vec{n}_\Gamma \jumpOp{ \fvor }
    - \frac{C}{h}\jumpOp{ \ftestLtwo } \cdot \mat{\mu} \cdot \jumpOp{ \fvor }
      \right]~dS
      \\
    +\int_\Gamma \jumpOp{ \ftestLtwo }
    \vec{n}_\Gamma \cdot \sigma_\| \avgOp{ \Vnabla_\| \fpot }
     ~dS
     ,
   \hspace{0.25in} 
  \forall \ftestLtwo \in \Ltwoq
   \label{eq:vorticityIdealWeakGen}
   .
\end{multline}
Here, we use the standard DG formulation of advection \cite{DGM_Brezi2004},
\cite{DGM2015} with the conservative numerical flux
$\hat{\fvor} = \avgOp{ \fvor } + \delta
\frac{|\vDrift\cdot\vec{n}_\Gamma|}{\vDrift\cdot\vec{n}_\Gamma}
\jumpOp{ \fvor }$.
A more complete discussion of the advection operator used by Eq. \eqref{eq:vorticityIdealWeakGen} can be found in 
\ref{app:DGformulation}.
We also use the standard DG formulation of diffusion \cite{Arnold2006}
where the parameter $C>0$ is the~DG diffusion interior penalty parameter and $h$ is the~cell size.
The parameter $\delta \in [0,1/2]$ controls
the amount of upwinding applied to the advection operator, where $\delta=0$ is the centered flux and $\delta=1/2$ is the fully upwinded flux.

\subsection{Enstrophy Conservation \label{sec:enstrophy}}

Solving either Eq. \eqref{eq:vorticityH1} or  Eq. \eqref{eq:vorticityIdealWeakGen} with a symplectic integrator will conserve enstrophy for an appropriate set of assumptions.

\begin{lemma}[Enstrophy as a discrete quadratic form]
  \label{lem:enstrophyQuadratic}
  Let the vorticity $\fvor$ be approximated within a finite-dimensional function space $\mathcal{V}$, e.g. either $\Ltwoq$ or $\Honeq$, spanned by a set of basis functions $\vect{\ftestgen}$, so that one can express the vorticity field as $\fvor = \vect{\ftestgen}^T\vect{\fvor}$, where $\vect{\fvor}$ represent the coefficients of the finite element representation.
  Then, the enstrophy is a quadratic form
  \begin{gather}
\Enstrophy=\tfrac{1}{2}    \int_\Omega w^2 ~dV =\tfrac{1}{2} \vect{\fvor}^T \matr{M} \vect{\fvor}
  \end{gather}
  where the mass matrix $ \matr{M}$ is defined by
   \begin{gather}
   \matr{M} = \int_\Omega \vect{\ftestgen} \vect{\ftestgen}^T ~dV.
     \end{gather}
\end{lemma}
\begin{proof}
  \begin{gather*}
    \int_\Omega \fvor^2 ~dV
      = \int_\Omega \vect{\ftestgen}^T \vect{\fvor} \vect{\ftestgen}^T \vect{\fvor} ~dV
      = \int_\Omega \vect{\fvor}^T \vect{\ftestgen} \vect{\ftestgen}^T \vect{\fvor}
      =  \vect{\fvor}^T \matr{M} \vect{\fvor}.
  \end{gather*}
\end{proof}

\begin{theorem}[Semi-discrete enstrophy conservation for continuous time] \label{lem:enstrophyConservedWeak}
  Let the velocity $\vDrift \in \Hdiv{}$ and assume
  (i) either $\fvor \in \Honeq$ in \eqref{eq:vorticityH1} or
  $\fvor \in \Ltwoq$ in \eqref{eq:vorticityIdealWeakGen} without upwinding
  ($\delta = 0$),
  (ii) $\Vnabla \cdot \vDrift = 0$ in $\Omega$, and
  (iii) the weak form satisfies $\int_{\partial \Omega}~\fvor^2~\vDrift \cdot~\vec{n}_{\partial \Omega}~dS~=~0$.
  Then, the spatially discretized vorticity ODE, which are continuous in time, conserve enstrophy
  \begin{gather*}
    \frac{d}{dt} \int_\Omega \fvor^2 ~dV = 0.
  \end{gather*}
  Note that the condition $\int_{\partial \Omega} \fvor^2 \vDrift \cdot \vec{n}_{\partial \Omega} ~dS = 0$ is met by periodic or homogeneous Dirichlet boundary conditions (or a mixture thereof) on $\fvor$ and $\vDrift \cdot \vec{n}_{\partial \Omega}$.
\end{theorem}
\begin{proof}
  Because the expression for $\fvor$ in \eqref{eq:vorticityIdealWeakGen} with $\delta = 0$ holds for all $\ftestLtwo \in \Ltwoq$ and $\fvor \in \Ltwoq$, the expression holds for $\ftestLtwo = \fvor$:
  \begin{gather*}
    \int_\Omega \fvor \partial_t \fvor ~dV =
      - \int_\Omega \fvor \vDrift \cdot \Vnabla \fvor ~dV
      + \int_{\Gamma} \avgOp{ \fvor } \jumpOp{ \fvor } 
        \vDrift \cdot \vec{n}_\Gamma ~dS
      .
  \end{gather*}
  Noting that $\avgOp{\fvor} \jumpOp{ \fvor } =
  \frac{1}{2} \jumpOp{ \fvor^2 }$ from \eqref{eq:productJump}
  and letting $a = \fvor^2$ in \eqref{eq:divTheoremDG}, the expression above
  yields
  \begin{gather*}
    \int_\Omega  \partial_t \fvor^2 ~dV =
      \int_\Omega \fvor^2 \Vnabla \cdot \vDrift ~dV
      - \int_{\partial \Omega} \fvor^2 \vDrift \cdot \vec{n}_{\partial \Omega} ~dS .
  \end{gather*}
  The assumptions $\Vnabla \cdot \vDrift = 0$ in $\Omega$ and $\int_{\partial \Omega} \fvor^2 \vDrift \cdot \vec{n}_{\partial \Omega} ~dS = 0$  then yield
  \begin{gather*}
    \frac{1}{2} \frac{d}{dt} \int_\Omega \fvor^2 ~dV = \int_\Omega \fvor \partial_t \fvor ~dV = 0,
  \end{gather*}
  which completes the proof.
\end{proof}

\begin{theorem}[Fully discrete enstrophy conservation]
  \label{thm:enstrophyConservedDiscrete}
  Let the velocity $\vDrift \in \Hdiv{}$ and assume
  (i) either $\fvor \in \Honeq$ in \eqref{eq:vorticityH1} or
  $\fvor \in \Ltwoq$ in \eqref{eq:vorticityIdealWeakGen} without upwinding ($\delta = 0$),
  (ii) $\Vnabla \cdot \vDrift = 0$ in $\Omega$, and
  (iii) periodic or homogeneous Dirichlet boundary conditions (or a mixture thereof). Then, if
  the discretized vorticity ODE are evolved in discrete time with a symplectic integrator,
  enstrophy will be conserved in discrete time.
\end{theorem}
\begin{proof}
  Because $\fvor \in \Ltwoq$, \Cref{lem:enstrophyQuadratic} implies that
  the enstrophy is equal to a quadratic form on $\Ltwoq$.
  Using (i) the weak formulations \eqref{eq:vorticityH1} or 
  \eqref{eq:vorticityIdealWeakGen}, in conjunction with conditions (ii)-(iii), 
  \Cref{lem:enstrophyConservedWeak} implies that the spatially discretized form 
  of the enstrophy is conserved in continuous time.
  If a symplectic integrator is used to integrate the resulting vorticity ODE, then \Cref{th:symplectic_integrators} implies that the enstrophy will be conserved in discrete time.
\end{proof}

\subsection{Energy Conservation \label{sec:energy}}
Integrating the discretized equations \eqref{eq:potentialWeak} with a symplectic integrator also conserves total kinetic energy  for an appropriate set of assumptions. In particular, the polynomial order of the finite element spaces for potential and vorticity must be consistent with one another.

\begin{lemma}[Energy as a discrete quadratic form for the potential]
  \label{lem:energyQuadratic}
  Let the potential $\fpot$ be approximated as an element of a finite-dimensional function space $\mathcal{V}$, e.g. either $\Honeq$ or $\Ltwoq$, and let $\vect{\ftestgen}$ be a vector of basis functions for this function space, so that the potential can be expressed as $\fpot = \vect{\ftestgen}^T \vect{\fpot}$.
  Then, the energy is equivalent to a quadratic form:
  \begin{gather*}
    \int_\Omega \frac{\rho}{\Bmag^2} \VnablaPerp \fpot \cdot \VnablaPerp \fpot ~dV = \vect{\fpot}^T \DiffusionMatrix \vect{\fpot},
  \end{gather*}
  where $\DiffusionMatrix$ is the symmetric positive semi-definite ``diffusion'' matrix defined via
  \begin{gather} \label{eq:diffusion_matrix}
    \DiffusionMatrixElements_{ij} = \int_\Omega \frac{\rho}{\Bmag^2} \VnablaPerp \ftestgen_i \cdot \VnablaPerp \ftestgen_j ~dV.
  \end{gather}
\end{lemma}
\begin{proof}
  \begin{align*}
    \int_\Omega \frac{\rho}{\Bmag^2} \VnablaPerp \fpot \cdot \VnablaPerp \fpot ~dV
      &= \int_\Omega \frac{\rho}{\Bmag^2} \VnablaPerp \left(\sum_i \fpot_i \ftestgen_i \right) \cdot \VnablaPerp \left(\sum_j \fpot_j \ftestgen_j\right) ~dV \\
      &= \sum_i \sum_j \fpot_i \fpot_j \int_\Omega \frac{\rho}{\Bmag^2} \VnablaPerp \ftestgen_i \cdot \VnablaPerp \ftestgen_j ~dV
      = \vect{\fpot}^T \DiffusionMatrix \vect{\fpot}
  \end{align*}
\end{proof}

\begin{corollary}[Energy as a discrete quadratic form for the vorticity]
  \label{cor:energyQuadraticVorticity}
  Suppose the potential approximation $\fpot=\vect{\ftestHone}^T \vect{\fpot}$ and vorticity approximation $\fvor=\vect{\ftestLtwo}^T \vect{\fvor}$ satisfy the weak formulation \eqref{eq:potentialWeak} with a set of boundary conditions that make $\fpot$ unique, or equivalently that
  \begin{gather*}
    \DiffusionMatrix \vect{\fpot} = \OverlapMatrix \vect{\fvor}
  \end{gather*}
  where
  \begin{align*}
     \OverlapMatrix = \int_\Omega \vect{\psi}\vect{\eta}^T ~dV
  \end{align*}
  and $\DiffusionMatrix$ the symmetric positive definite matrix in Eq. \eqref{eq:diffusion_matrix} defined by the the choice of $\mathcal{V}$.
  Then energy can be expressed as a quadratic form for the vorticity:
  \begin{gather*}
    \Energy =\tfrac{1}{2}  \vect{\fpot}^T \DiffusionMatrix \vect{\fpot} =\tfrac{1}{2} \vect{\fvor}^T \OverlapMatrix^T \DiffusionMatrix^{-1}  \OverlapMatrix \vect{\fvor}.
  \end{gather*}
\end{corollary}

\begin{theorem}[Semi-discrete energy conservation for continuous time]
  \label{lem:energyConservedWeak}
   Let the potential $\fpot \in \Honep$ determined by \eqref{eq:potentialWeak} 
   and velocity $\vDrift \in \Hdiv{}$
   determined by \eqref{eq:driftV}, and assume
  (i) either $\fvor \in \Honeq$ in \eqref{eq:vorticityH1} or
  $\fvor \in \Ltwoq$ in \eqref{eq:vorticityIdealWeakGen},
  where $\qvor \geq \ppot$,
  (ii) $\Vnabla \cdot \vDrift = 0$ in $\Omega$,
  (iii) $\int_{\partial\Omega} \fvor \fpot \vDrift \cdot \vec{n}_{\partial \Omega} ~dS = 0$,
  (iv) $\int_{\partial \Omega} \fpot \VnablaPerp \fpot \cdot \vec{n}_{\partial \Omega} ~dS = 0$, and
  (v) $\int_{\partial \Omega} \partial_t \fpot \VnablaPerp \fpot \cdot \vec{n}_{\partial \Omega} ~dS = 0$.
  Then, the spatially discretized ODEs, which are continuous in time, 
  conserve energy,
  \begin{gather*}
    \frac{d}{dt} \int_\Omega \frac{\rho}{\Bmag^2} \VnablaPerp \fpot \cdot \VnablaPerp \fpot ~dV = 0.
  \end{gather*}
  Note that the conditions $\int_{\partial\Omega} \fvor \fpot \vDrift \cdot \vec{n}_{\partial \Omega} ~dS = 0$, $\int_{\partial \Omega} \fpot \VnablaPerp \fpot \cdot \vec{n}_{\partial \Omega} ~dS = 0$, and $\int_{\partial \Omega} \partial_t \fpot \VnablaPerp \fpot \cdot \vec{n}_{\partial \Omega} ~dS = 0$ are met by periodic or homogeneous Dirichlet boundary conditions (or a mixture thereof) on $\fvor$ and $\fpot$.
\end{theorem}
\begin{proof}
Using \eqref{eq:potentialWeak} with the test function $\ftestHone=\fpot  \in \Honep$ along with (iv) $\int_{\partial \Omega} \fpot \VnablaPerp \fpot \cdot \vec{n}_{\partial \Omega} ~dS = 0$ implies
  \begin{gather*}
      \int_\Omega \frac{\rho}{\Bmag^2} \VnablaPerp \fpot  \cdot \VnablaPerp \fpot ~dV =  \int_\Omega \fvor \fpot ~dV
      .
  \end{gather*}
  The time derivative of this expression is
  \begin{gather*}
       \frac{d}{dt} \int_\Omega \frac{\rho}{\Bmag^2} \VnablaPerp \fpot  \cdot \VnablaPerp \fpot ~dV = \frac{d}{dt} \int_\Omega \fvor \fpot ~dV
       =   \int_\Omega \partial_t\left(\fpot  \fvor   \right)~dV
      .
  \end{gather*}
Using \eqref{eq:potentialWeak} with the test function $\ftestHone=\partial_t \fpot  \in \Honep$ along with (v) $\int_{\partial \Omega} \partial_t \fpot \VnablaPerp \fpot \cdot \vec{n}_{\partial \Omega} ~dS = 0$ yields
  \begin{gather*}
    \frac{1}{2} \frac{d}{dt} \int_\Omega \frac{\rho}{\Bmag^2} \VnablaPerp \fpot \cdot \VnablaPerp \fpot ~dV = \int_\Omega \fvor \partial_t \fpot ~dV,
  \end{gather*}
  because $\int_{\partial \Omega} \partial_t \fpot \VnablaPerp \fpot \cdot \vec{n}_{\partial \Omega} ~dS = 0$ by assumption.
  Subtracting the two results yields
  \begin{gather*}
    \frac{1}{2} \frac{d}{dt} \int_\Omega \frac{\rho}{\Bmag^2} \VnablaPerp \fpot \cdot \VnablaPerp \fpot ~dV = \int_\Omega \fpot \partial_t \fvor ~dV.
  \end{gather*}

  The DG formulation for vorticity \eqref{eq:vorticityIdealWeakGen} holds for 
  any $\ftestLtwo \in \Ltwoq$ and $\fpot \in \Honep$. Crucially, since 
  $\Honep \subset \Ltwoq$, as long as $\qvor \geq \ppot$, then the following 
  expression holds for $\ftestLtwo = \fpot$:
  \begin{gather*}
    \int_\Omega \fpot \partial_t \fvor ~dV =
      - \int_\Omega \fpot \vDrift \cdot \Vnabla \fvor ~dV
      + \int_{\Gamma} \avgOp{ \fpot } \jumpOp{ \fvor } \vDrift \cdot  \vec{n}_\Gamma~dS
      - \delta \int_{\Gamma} |\vDrift\cdot \vec{n}_\Gamma|
        \jumpOp{ \fvor } \jumpOp{ \fpot } ~dS.
  \end{gather*}
  Note that because of the~continuity of $\fpot \in \Honep$,
  one has $\jumpOp{ \fpot } = 0$, and the upwinding term vanishes for any
  $\delta$. Choosing
  $a = \fpot  \fvor$ in \eqref{eq:divTheoremDG} and $a = \fpot, b = \fvor$ in
  \ref{eq:productJump}, and without the jump terms involving
  $\jumpOp{ \fpot }$, the above expression simplifies to
  \begin{gather*}
    \int_\Omega \fpot \partial_t \fvor ~dV =
      \int_\Omega \fvor \vDrift \cdot \Vnabla \fpot ~dV
      + \int_{\Omega}
        \fpot \fvor \Vnabla \cdot \vDrift ~dV
      - \int_{\partial\Omega} \fpot \fvor
        \vDrift \cdot \vec{n}_{\partial \Omega} ~dS.
  \end{gather*}
  The continuous formulation in \eqref{eq:vorticityH1}, where $\fvor, \ftestHone\in \Honeq$, leads to the same result, since the jumps, $\jumpOp{\fvor}$, vanish from the outset. Once again, the result holds as long as $\Honep\subset\Honeq$, which simply requires $\qvor\geq \ppot$.

  Using assumptions (ii) $\Vnabla \cdot \vDrift = 0$ and (iii)
  $\int_{\partial\Omega} \fvor \fpot \vDrift \cdot \vec{n}_{\partial \Omega} ~dS = 0$, and noting that \eqref{eq:driftV} implies that $\vDrift\cdot \Vnabla \phi =0$  completes the proof:
  \begin{gather*}
    \frac{1}{2} \frac{d}{dt} \int_\Omega \frac{\rho}{\Bmag^2} \VnablaPerp \fpot \cdot \VnablaPerp \fpot ~dV = \int_\Omega \fpot \partial_t \fvor ~dV = 0  .
  \end{gather*}
\end{proof}

\begin{theorem}[Fully discrete energy conservation]
  \label{thm:energyConservedDiscrete} 
  Let the potential $\fpot \in \Honep$ determined by \eqref{eq:potentialWeak} 
  with a set of boundary conditions that make $\fpot$ unique
  and velocity $\vDrift \in \Hdiv{}$
  determined by \eqref{eq:driftV}, and assume
  (i) either $\fvor \in \Honeq$ in \eqref{eq:vorticityH1} or
  $\fvor \in \Ltwoq$ in \eqref{eq:vorticityIdealWeakGen} where 
  $\qvor \geq \ppot$,
  (ii) $\Vnabla \cdot \vDrift = 0$ in $\Omega$, and
  (iii) periodic or homogeneous Dirichlet boundary conditions (or a mixture thereof).
  Then if the spatially discretized ODEs are evolved in discrete time with a symplectic integrator,
  energy will be conserved in discrete time.
\end{theorem}
\begin{proof}
  Because $\fpot \in \Honep$, \Cref{cor:energyQuadraticVorticity} implies that energy is equal to a quadratic form
  of the time-dependent variables, $\vect{w}$.
 Using (i) the weak formulation of \eqref{eq:vorticityIdealWeakGen} and the assumptions (ii)-(iiii), \Cref{lem:energyConservedWeak} implies that the energy is conserved in continuous time.
 If a symplectic integrator is used to integrate the resulting set of ODEs, then \Cref{th:symplectic_integrators} implies that the energy will be conserved in discrete time.
\end{proof}

\subsection{Poisson Bracket Formulation \label{sec:poisson_bracket} }

The results of the previous sections can also be naturally stated in a manner that expresses the advection operator as a Poisson bracket.
This is useful for connecting with the extensive theoretical and numerical literature that utilizes the Poisson bracket to express the advection operator.
This subsection shows that these theorems explicitly hold true for the Poisson bracket formulation.

The~particular case of constant $\VB=\Bmag_0\zhat$ for constant $\Bmag_0$, in
\eqref{eq:model} corresponds to an incompressible neutral fluid model in the $x$-$y$ plane.
The~drift velocity reduces to
$ \vDriftxy =\zhat\times\Vnabla \phi /\Bmag_0=  [-\pdv{\fpot}{y}, \pdv{\fpot}{x}]/\Bmag_0$.
The definition of the~Poisson bracket
\begin{align}
  \PoissonBracket{f}{g} = \pdv{f}{x}\pdv{g}{y} - \pdv{f}{y}\pdv{g}{x}
\end{align}
implies that one can write the~general vorticity equation
\eqref{eq:vorticityIdealWeakGen} as
\begin{multline}
  \int_\Omega \ftestLtwo \partial_t \fvor ~dV =
    - \int_\Omega  \PoissonBracket{\fpot}{\fvor} \ftestLtwo ~dV/\Bmag_0
    + \int_{\Gamma} \avgOp{ \ftestLtwo }
      \jumpOp{ \fvor } \vDriftxy \cdot \vec{n}_\Gamma~dS
  \\
    - \delta \int_{\Gamma} \jumpOp{ \ftestLtwo }  \jumpOp{ \fvor } \abs{\vDriftxy  \cdot \vec{n}_\Gamma}~dS
  , \hspace{0.25in} ~\forall \ftestLtwo \in \Ltwoq
  , \label{eq:vorticityIdealWeakGenPoisson}
\end{multline}

In order to give the Poisson bracket form of the conservation laws, simply use
\eqref{eq:vorticityIdealWeakGenPoisson} with  appropriate choices of
test functions $\ftestLtwo$ to find:
\begin{align}
  \int_\Omega \partial_t \fvor ~dV &=
    - \int_\Omega  \PoissonBracket{\fpot}{\fvor}  ~dV/\Bmag_0
    + \int_{\Gamma}  \jumpOp{ \fvor } \vDriftxy \cdot \vec{n}_\Gamma  ~dS ,
  \label{eq:vortConservePoisson} \\
  \int_\Omega \fpot \partial_t \fvor ~dV &=
    - \int_\Omega  \PoissonBracket{ \tfrac{1}{2} \fpot^2}{\fvor}   ~dV/\Bmag_0
    + \int_{\Gamma}  \jumpOp{ \fpot \fvor } \vDriftxy \cdot \vec{n}_\Gamma ~dS ,
  \label{eq:energyConservePoisson} \\
  \int_\Omega  \partial_t \fvor^2 ~dV &=
    - \int_\Omega  \PoissonBracket{\fpot}{  \fvor^2} ~dV/\Bmag_0
    + \int_{\Gamma}  \jumpOp{ \fvor^2 } \vDriftxy \cdot \vec{n}_\Gamma ~dS
    - 2\delta \int_{\Gamma}
      \jumpOp{ \fvor }^2 \abs{\vDriftxy  \cdot \vec{n}_\Gamma}~dS
  \label{eq:enstrophyConservePoisson}.
\end{align}
Vorticity conservation uses $\ftestLtwo = 1$ in
\eqref{eq:vortConservePoisson}, energy conservation uses
$\ftestLtwo = \fpot$ in \eqref{eq:energyConservePoisson}, and enstrophy
conservation uses $\ftestLtwo = \fvor$ in \eqref{eq:enstrophyConservePoisson}.
The proof simply requires recalling the facts that $\Vnabla \cdot \vDriftxy = 0$ and that
$\vDriftxy \in \Hdiv{} $ implies $\jumpOp{ \vDriftxy } \cdot \vec{n}_\Gamma= 0$.

\section{Numerical Results}\label{sec:results}

\subsection{Kelvin-Helmholtz Instability \label{sec:kh_instability} }
The Kelvin-Helmholtz (KH) instability is a well-known fluid instability that occurs when nontrivial shear is present in the velocity field. This instability begins with the linear growth of unstable modes, then transitions to a non-linear stage of turbulence, and concludes with turbulence decay. To simulate the KH instability, Eq.~\eqref{eq:NS_vorticity} with $\alpha = S_w = 0$ is solved using the FEEC formulation detailed in \Cref{sec:method}. A conjugate gradient linear solver with the BoomerAMG preconditioner is used for \eqref{eq:potentialWeak}, 
a GMRES linear solver with AIR preconditioner \cite{AIR2018,AIR2019} 
is used for \eqref{eq:vorticityH1}, and a GMRES linear solver with either ILU 
preconditioner is used for \eqref{eq:vorticityIdealWeakGen}. 
Piccard iteration is used as nonlinear solver. Results from the linear and 
nonlinear phases are presented in the two subsections below.

\subsubsection{Geometry and initial conditions}
The physical domain used for the KH instability is a 2D rectangular slab with lengths $L_x$ and $L_y$ in the $x$ and $y$ directions, respectively. The fluid flows in the $y$ direction with a sinusoidal profile given by the following
\begin{align}
  \fpot_{0} &{=} \sin(mk_{x0} x)/m^2,
  \nonumber \\
  \vec v_{\perp,0} &{=} V_0 \cos(mk_{x0} x)\hat{y},
  \label{eq:ic} \\
  \fvor_0 &{=} - k_{x0}^2 \sin(mk_{x0} x) ,
  \nonumber
\end{align}
where $k_{x0}=2\pi/L_x$ and $V_0=k_{x0}/m$. To induce the instability, a perturbation is added to the initial conditions. For example, the full initial vorticity is given by 
$ \fvor(t=0) = \fvor_0 + \tilde{\fvor}_0$, where $\tilde{\fvor}_0$ is a perturbation defined by
\begin{align}
  \tilde{\fvor}_0 =
    \epsilon~ \fvor_0  \sum_{n=1}^{n=N} e^{-(n k_{y0} / N)^2} \sin{(n k_{y0}y)} .
    \label{eq:wPerturbation}
\end{align}
In the above, $\epsilon \in[10^{-8},10^{-2}]$ and $k_{y0}=2\pi/L_y$. In this scenario, the perturbation will induce the KH instability whenever $k_{y0} {\leq} mk_{x0}$, or, $L_y>mL_x$.
 
\subsubsection{ Verification of the linear growth rate}
Verification of the numerical implementation of the governing equations is performed by comparing the computed linear growth rates against quasi-analytical predictions from linear stability analysis. That is, assume $\fpot = \fpot_0 + \fpot_1$ and $\fvor = \fvor_0 + \fvor_1$, where $\fpot_0$ and $\fvor_0$ are the equilibrium solutions and $\fpot_1$ and $\fvor_1$ are the perturbed components given by
\begin{equation}
  \fpot_1= \fpot_1(x)e^{\gamma t+ik_yy},
\end{equation}
\begin{equation}
  \fvor_1=(\fpot_1''-k_y^2\fpot_1)e^{\gamma t+ik_yy}.
\end{equation}
The vorticity equation can be linearized by neglecting high-order nonlinear terms, which results in the 1D eigenvalue equation
\begin{equation}
    \gamma(\fpot_1''-k_y^2\fpot_1)=k_y[\fpot_0'(\fpot_1''-k_y^2\fpot_1)-\fpot_0'''\fpot_1],
\end{equation}
where $\fpot'=d\fpot/dx$. After specifying the boundary conditions for the $x$ direction, one obtains a 1D eigenvalue problem that can be solved numerically.
The numerically simulated linear growth rates are calculated from the finite element simulation by first performing a spatial Fourier transform to determine $\fpot_k$, where $k$ is the wavenumber, and then calculating the growth rate 
\begin{equation*}
    \gamma_{\rm sim} =\frac{1}{2\Delta t}\ln\left(\frac{|\fpot_k|_{t+\Delta t}}{|\fpot_k|_{t-\Delta t}}\right)
\end{equation*}
for each Fourier mode.

For this test case we choose $m=1, L_x=1$, and $L_y=16$. A grid with resolution of $32 \times  512$ second-order elements is used. Simulations are carried out with two different boundary conditions along the $x$ direction: Dirichlet ($\fpot(x=0,L_x)=0$) and periodic. Periodic boundaries are always used in the $y$ direction.
Figure~\ref{fig:khi_lgr} summarizes the MFEM verification results. Two combinations of spatial and temporal discretizations are employed: (1) MFEM-H1 using continuous spatial elements for potential $\fpot \in \Hone{p}$ \eqref{eq:potentialWeak} and vorticity $\fvor \in \Hone{p}$ \eqref{eq:vorticityH1} and the backward Euler temporal integration and (2) MFEM-DG which uses $\fpot \in \Hone{p}$ and spatial 
DG elements for vorticity $\fvor \in \Ltwo{p}$ \eqref{eq:vorticityIdealWeakGen} and the second-order implicit midpoint temporal integration. As the results from these tests show, both approaches yield good agreement with the 1D eigensolver results for different $x$ boundary conditions. This is as expected, since the nonlinear interaction between modes is at minimum during the linear stage and the numerical dispersion level is relatively low. As a result, this particular case does not yet highlight the advantages of the energy- and enstrophy-conserving porperties of MFEM-DG.

\begin{figure}[h]
  \centering
  \includegraphics[width=0.6\textwidth]{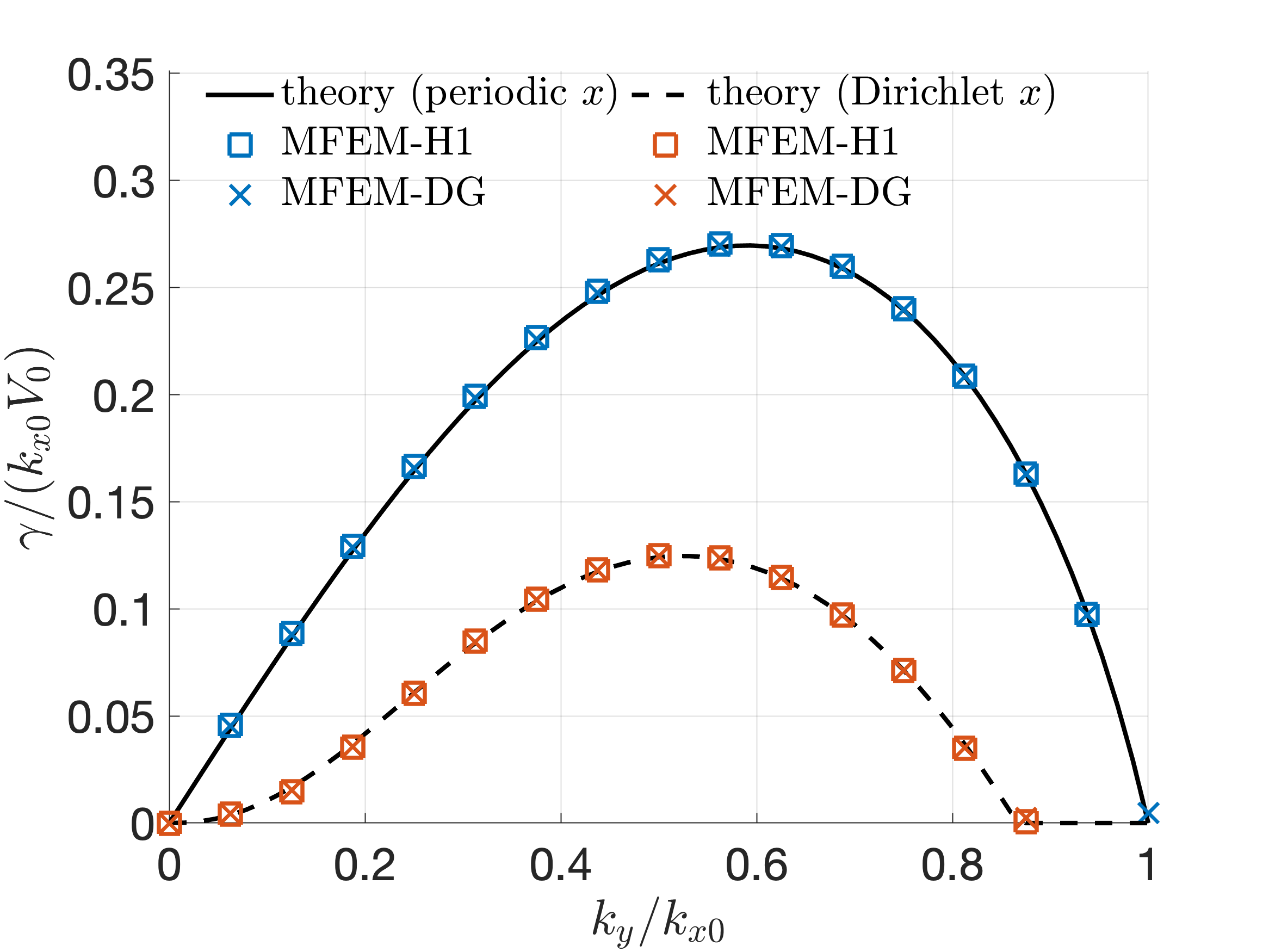}%
  \caption{MFEM simulated linear growth rates of Kelvin-Helmholtz instability v.s. theoretical prediction using $p=2$. }%
   \label{fig:khi_lgr}
\end{figure}

The impact of the polynomial order $p$ on the determination of the growth rates is investigated next. As expected, higher order elements are beneficial for accurately capturing the instability 
when the degrees of freedom (DOFs) is relatively small. If the DOFs is relatively high, e.g, DOFs$_\varphi=64\times 1024$ in Figure.~\ref{fig:khi_lgr}, both low- and high-order polynomials yield quantitatively the same linear growth rates.
However, as shown in Figure~\ref{fig:khi_lgr2}, if the DOFs is relatively low ($\sim8\times 128$ for $\varphi$), low-order polynomials (e.g, first- and second-order) systematically overestimate the linear growth rate at $ky/k_{x0}\approx 0$ and $1$; while the high-order elements (third-order) do a better job of accurately recovering the linear growth rates. 
Note that, because the simulation is nonlinear and multiple wave numbers are initially seeded, it can be difficult to accurately determine the growth rate when it is small.

\begin{figure}[h]
  \centering
  \includegraphics[width=0.6\linewidth]{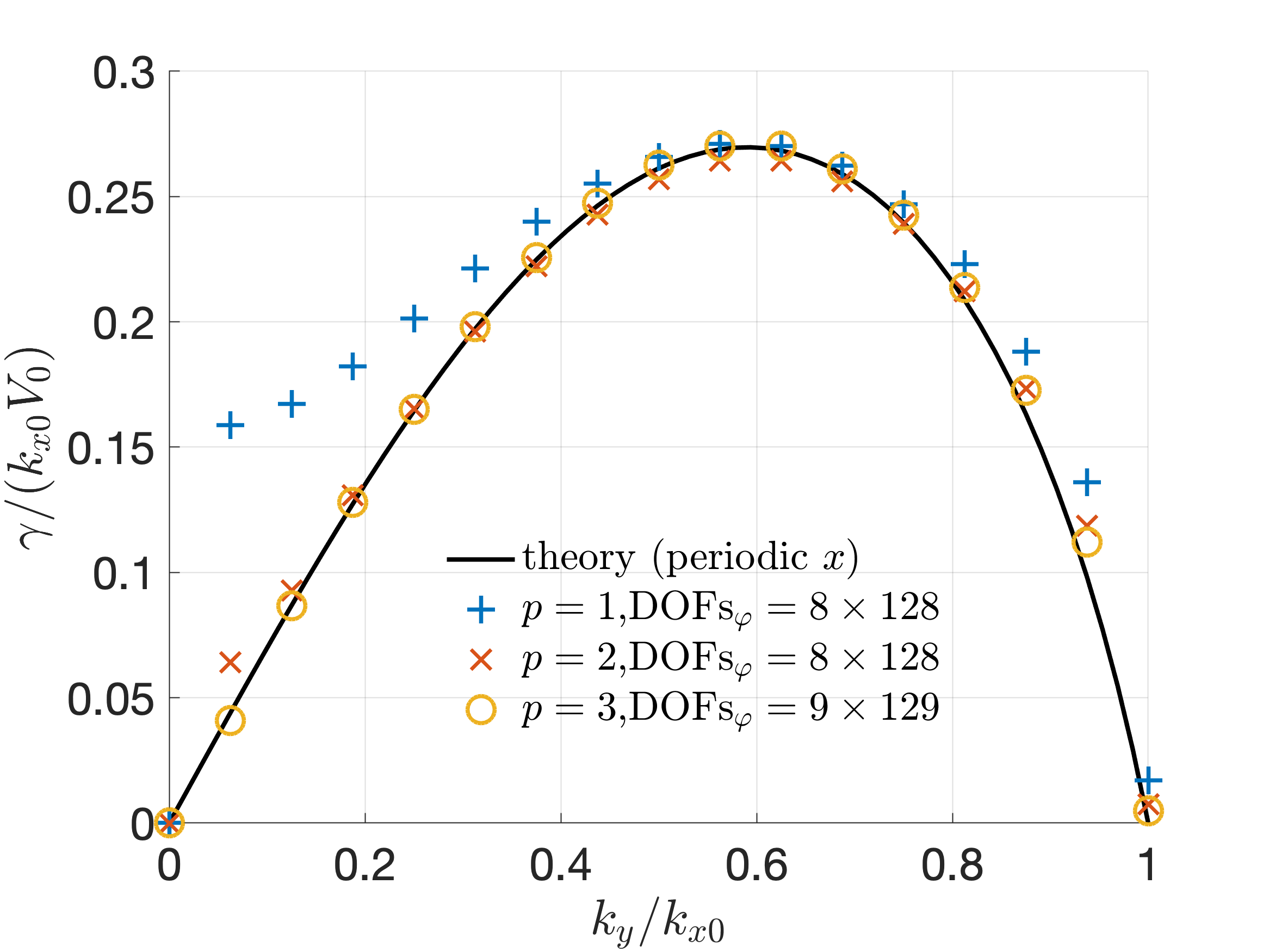}%
  \caption{Effect of order of elements $p$ on linear growth rates for the KHI instability when the DOFs is relatively small.} %
   \label{fig:khi_lgr2}
\end{figure}

\subsubsection{Verification of nonlinear energy \& enstrophy conservation  }
 In this subsection,  the viscosity, friction and forcing are all set to zero so that the conservation of energy and enstrophy, which applies only in the non-dissipative and unforced case, can be verified.
All simulations use the parameters $L_x = 0.5$, $L_y = 2.0$, and $t_{final} = 2.0$. The values $m = 5$, $N = 100$, $\epsilon = 10^{-2}$ are used for the initial conditions given by \eqref{eq:ic} and \eqref{eq:wPerturbation}. 
 The mesh consists of $192 \times 48$ rectangular elements with cubic polynomials, and periodic boundary conditions are enforced. 
Piccard iteration and linear solvers use relative tolerance $10^{-12}$.


\def\fs{0.65}
\begin{figure}
  \begin{subfigure}{\textwidth}
    \centering
    \includegraphics[width=\fs\textwidth]{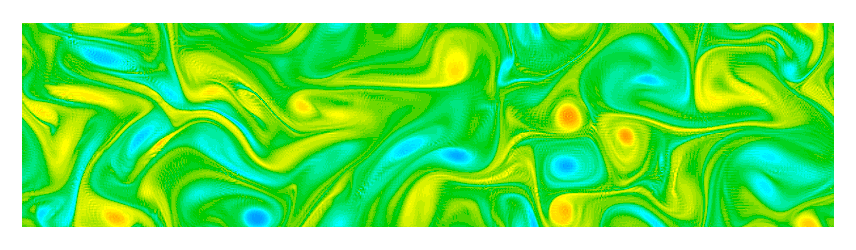}
    \caption{Does not conserve energy or enstrophy}
    \label{fig:khi_spatial_structure_a}
  \end{subfigure}

  \begin{subfigure}{\textwidth}
    \centering
    \includegraphics[width=\fs\textwidth]{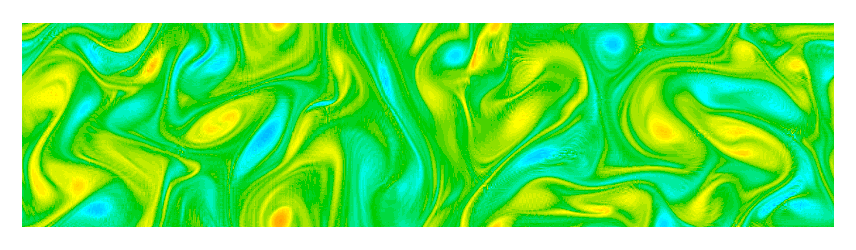}
    \caption{Does not conserve energy or enstrophy}
    \label{fig:khi_spatial_structure_b}
  \end{subfigure}

  \begin{subfigure}{\textwidth}
    \centering
    \includegraphics[width=\fs\textwidth]{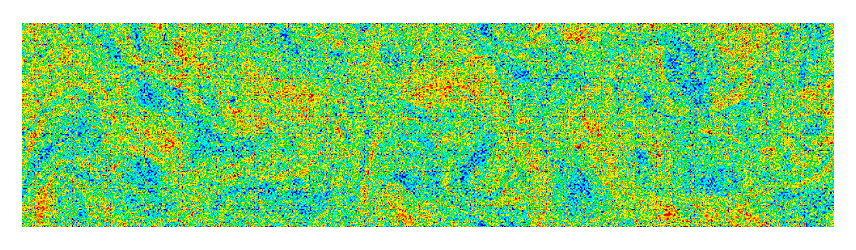}
    \caption{Conserves energy and enstrophy}
    \label{fig:khi_spatial_structure_c}
  \end{subfigure}

  \begin{subfigure}{\textwidth}
    \centering
    \includegraphics[width=\fs\textwidth]{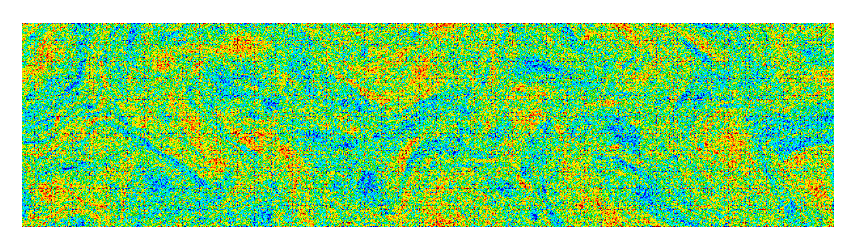}
    \caption{Conserves energy and enstrophy}
    \label{fig:khi_spatial_structure_d}
  \end{subfigure}

  \begin{subfigure}{\textwidth}
    \centering
    \includegraphics[width=\fs\textwidth]{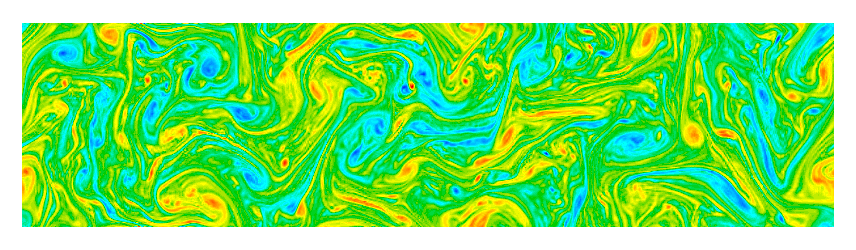}
    \caption{Approximately conserves energy, decays enstrophy}
    \label{fig:khi_spatial_structure_e}
  \end{subfigure}

  \begin{subfigure}{\textwidth}
    \centering
    \includegraphics[width=\fs\textwidth]{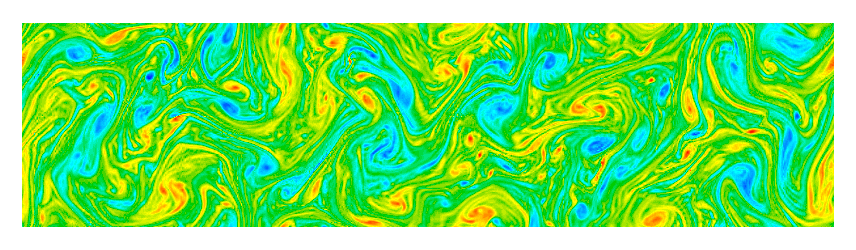}
    \caption{Conserves energy, decays enstrophy}
    \label{fig:khi_spatial_structure_f}
  \end{subfigure}

  \caption{ Vorticity at $t=1.0$ for the conservative model (no dissipation or forcing) for polynomial order $3$. Different space-time discretization schemes: 
  (a) H1 / backward Euler, 
  (b) DG / backward Euler with $\delta=0$, 
  (c) H1 / implicit midpoint, 
  (d) DG / implicit midpoint with $\delta=0$, 
  (e) DG$_{\delta=0.05}$ / explicit midpoint, and
  (f) DG$_{\delta=0.05}$ / implicit midpoint. 
   All figures use an identical color scale; this emphasizes the long wavelength structures in (c) and (d) which would otherwise be hidden by scaling to the extreme maxima and minima of these cases. 
  }
  \label{fig:khi_spatial_structure}
\end{figure}

The vorticity field is shown in Fig.~\ref{fig:khi_spatial_structure} at $t = 1.0$ for different spatial-temporal discretizations. 
Figures~\ref{fig:khi_spatial_structure_a} and \ref{fig:khi_spatial_structure_b} employ the backward-Euler time discretization, with the former having $\fvor \in \Hone{p}$ and the latter $\fvor \in \Ltwo{p}$. 
Figures~\ref{fig:khi_spatial_structure_c} and \ref{fig:khi_spatial_structure_d} employ the implicit midpoint time discretization, again with the former using $\fvor \in \Hone{p}$ and the latter $\fvor \in \Ltwo{p}$. 
Finally, Figs.~\ref{fig:khi_spatial_structure_e} and \ref{fig:khi_spatial_structure_f} use the explicit midpoint and implicit midpoint discretizations, respectively, but both use upwinding in the DG spatial discretization.

As Figs.~\ref{fig:khi_spatial_structure_a} and \ref{fig:khi_spatial_structure_b} show, the backward-Euler time discretization introduces a significant amount of dissipation since small-scale features are absent from the vorticity contours, as compared with the other spatial-temporal discretizations. On the other hand, implicit midpoint time discretization leads to a highly oscillatory solution for the small scales for both $\fvor \in \Hone{p}$ and $\fvor \in \Ltwo{p}$, as shown in Figs.~\ref{fig:khi_spatial_structure_c} and \ref{fig:khi_spatial_structure_d}.
This behavior cannot be alleviated by shrinking the time step, because 
the oscillations relate to the finite spatial resolution and consequently the~limited size of the~smallest vortices in the turbulence.
 Introducing upwinding in the DG scheme and using either explicit midpoint (Fig.~\ref{fig:khi_spatial_structure_e}) or implicit midpoint (Fig.~\ref{fig:khi_spatial_structure_f}) leads to the expected broad range of turbulent scales, with more fine scale structure than in the cases with backward-Euler time discretization, but without the mesh-scale fluctuations seen in the cases with implicit midpoint time discretization. Even small values of upwinding, such as $\delta = 0.05$, effectively dissipate nonphysical small-scale oscillations.

\def\fs{0.5}
\begin{figure}
  \begin{subfigure}{\fs\textwidth}
    \centering
    \includegraphics[width=\textwidth]{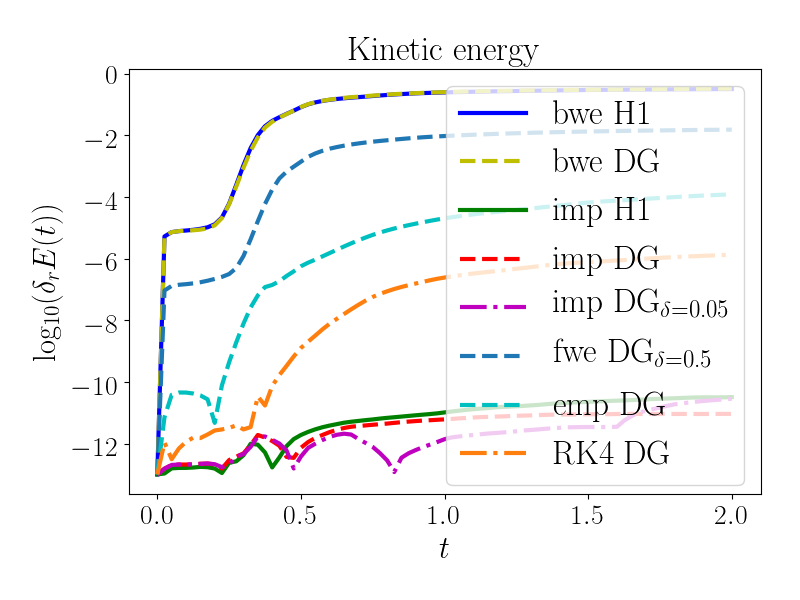}
    \caption{}
    \label{fig:khi_ee_conservation_a}
  \end{subfigure}
  \begin{subfigure}{\fs\textwidth}
    \centering
    \includegraphics[width=\textwidth]{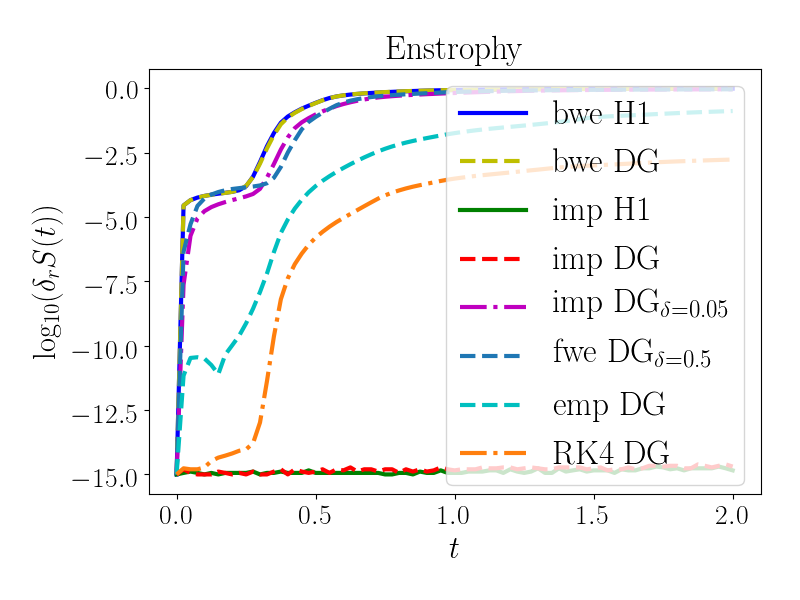}
    \caption{}
    \label{fig:khi_ee_conservation_b}
  \end{subfigure}

  \caption{
    Energy and enstrophy conservation for multiple spatial-temporal
    discretization using polynomial order $3$.}
  \label{fig:khi_ee_conservation}
\end{figure}
Figure~\ref{fig:khi_ee_conservation} shows the extent to which kinetic energy and enstrophy are conserved when using different spatial-temporal discretizations. As previously mentioned, the backward-Euler time discretization ({\it bwe H1} and {\it bwe DG} in Fig.~\ref{fig:khi_ee_conservation}) does not preserve the quadratic invariant \eqref{eq:quad}, and thus neither energy \Cref{thm:energyConservedDiscrete} nor enstrophy \Cref{thm:enstrophyConservedDiscrete} applies. As Fig.~\ref{fig:khi_ee_conservation} shows, kinetic energy and enstrophy are conserved when using the implicit midpoint integrator, either for $\fvor \in \Hone{p}$ ({\it imp H1}) or $\fvor \in \Ltwo{p}$ ({\it imp DG}), when no upwinding is used for DG. If upwinding is used, then kinetic energy is still conserved, but enstrophy is no longer conserved. This is in agreement with both \Cref{thm:enstrophyConservedDiscrete} and \Cref{thm:energyConservedDiscrete}. Neither energy nor enstrophy are conserved when using the forward Euler ({\it fwe DG}) time discretization, which uses upwinding $\delta = 0.5$. For this discretization, no upwinding or small values of upwinding lead to unstable simulations. We note that runs were also obtained, though not shown in Fig.~\ref{fig:khi_ee_conservation}, with the RK4 integrator, which were stable without any upwinding. The explicit midpoint (emp DG) scheme conserves enstrophy and energy well when not using upwinding. When upwinding is employed, the kinetic energy is still conserved, but enstrophy is no longer conserved. The surprisingly good conservation properties of the explicit midpoint scheme likely comes from its ability to  approximate implicit midpoint combined with the small size of the explicit time step, as discussed in \Cref{sec:explicit}.

\begin{figure}
  \centering
  \includegraphics[width=\fs\textwidth]{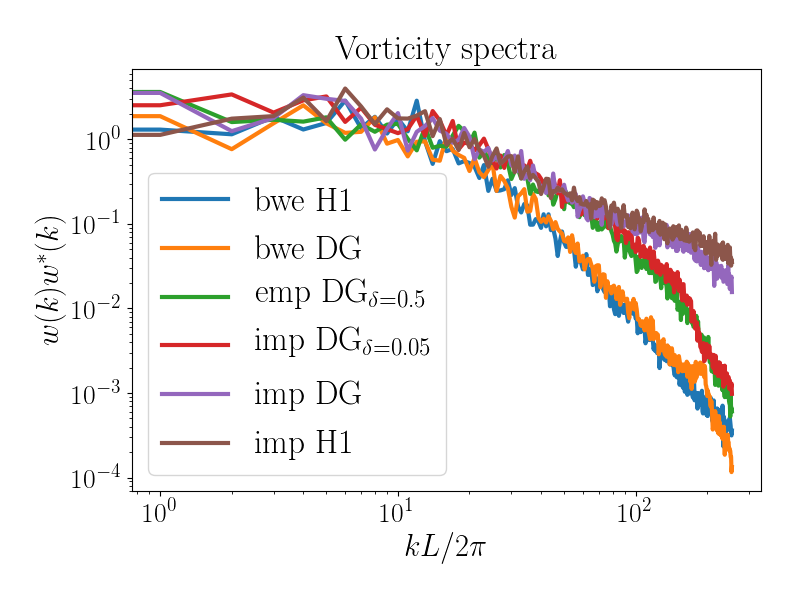}
  \caption{The Fourier spectrum of vorticity vs. $k_y$ at constant $k_x=5$, averaged over $t \in [0.6, 1.0]$. Various time-space discretizations are used: forward and backward Euler, explicit and implicit midpoint, and H1 and DG with and without upwinding for polynomial order 3.}
  \label{fig:vorticity_spectrum_all}
\end{figure}

\def\fs{0.5}
\begin{figure}
  \begin{subfigure}{\fs\textwidth}
    \centering
    \includegraphics[width=\textwidth]{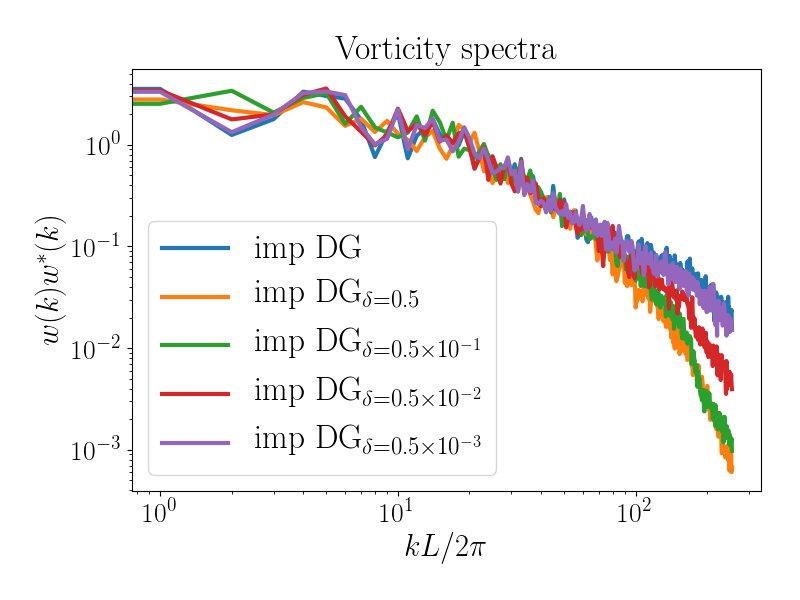}
    \caption{}
    \label{fig:vorticity_spectrum_DGscan_a}
  \end{subfigure}
  \begin{subfigure}{\fs\textwidth}
    \centering
    \includegraphics[width=\textwidth]{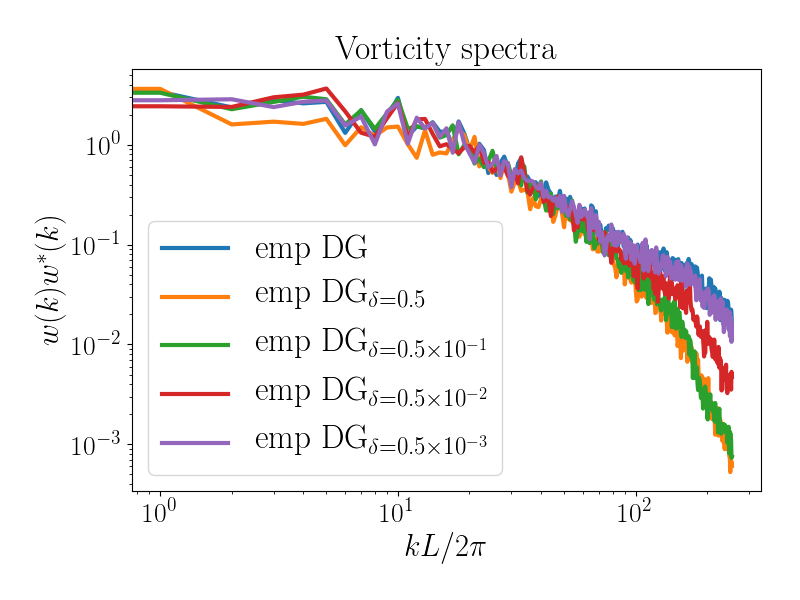}
    \caption{}
    \label{fig:vorticity_spectrum_DGscan_b}
  \end{subfigure}

  \caption{The Fourier spectrum of vorticity vs. $k_y$ at constant $k_x=5$, averaged over $t \in [0.6, 1.0]$. The effect of different levels of upwinding for explicit and implicit midpoint time discretization are shown. Energy is conserved to machine precision, while enstrophy is only conserved for implicit midpoint without upwinding (imp DG which has $\delta=0$).}
  \label{fig:vorticity_spectrum_DGscan}
\end{figure}

Figure~\ref{fig:vorticity_spectrum_all} examines a quantitative metric for the different behaviors of vorticity that follow from using the different spatial-temporal discretizations mentioned thus far. The backward-Euler cases lead to significant suppression of the vorticity for mid and high wave numbers. On the other hand, the implicit midpoint cases imp H1 and imp DG produce excessive amounts of vorticity for high wave numbers, which corresponds to the nonphysical small scale oscillations shown in Figs.~\ref{fig:khi_spatial_structure_c} and \ref{fig:khi_spatial_structure_d}. As previously mentioned, these schemes do not inject any dissipation. Finally, DG schemes that introduce dissipation through non-zero values of upwinding, whether they be paired with explicit or implicit midpoint integrators, provide what appear to be the best vorticity spectra. This judgement is based on the observation that the spectra in these cases track those from the dissipation-free implicit mid-point cases for low and intermediate wavenumbers, and fall off at high wavenumbers sufficiently to avoid the mesh-scale fluctuations seen in the dissipation-free cases.
We note that the simulation results differ only slightly when using explicit vs. implicit midpoint time discretizations and that the small-scale turbulence features are mainly determined by the level of upwinding $\delta$. This is shown in Fig.~\ref{fig:vorticity_spectrum_DGscan}, where it can be seen that the spectral representation of vorticity is almost identical when using the explicit and implicit midpoint schemes.

Finally, \Cref{tab:wallClockTime} shows the wall-clock times for simulations using various discretizations. Implicit midpoint DG with $\delta = 0.05$, which conserves energy and suppresses nonphysical small scale oscillations through upwinding, exhibits the shortest computational time. The backward Euler scheme without dissipation is slightly slower, but it does not conserve either energy nor enstrophy. The explicit midpoint integrator conserves energy appropriately but for this specific case it is roughly eleven times slower than the implicit midpoint scheme. The RK4 integrator, which has not been described thus far, shows better
conservation of energy than explicit midpoint while taking similar wall-clock time.

\def\fs{1.0}
\begin{table}[h]
  \resizebox{\textwidth}{!}{%
    \begin{tabular}{cccccc}
      & imp DG$^{\Delta t = 10^{-3}}_{\delta = 0.05}$ 
      & bwe DG$^{\Delta t = 10^{-3}}$ 
      & fwe DG$^{\Delta t = 1.5 \times 10^{-5}}_{\delta = 0.5}$ 
      & emp DG$^{\Delta t = 3 \times 10^{-5}}_{\delta = 0.05}$ 
      & RK4 DG$^{\Delta t = 6 \times 10^{-5}}_{\delta = 0.05}$ \\
      \hhline{======}
      time & 2425 s & 2861 s & 28694 s & 27291 s & 25612 s \\
      {\bf factor} & {\bf 1.0} & {\bf 1.2} & {\bf 11.0} & {\bf 11.3} & {\bf 10.6}
    \end{tabular}
  }
  \caption{Wall-clock times and corresponding slow down for different 
    spatial-temporal discretizations of the 2D incompressible Navier-Stokes 
    equations. Since the mesh employed consists of $192\times48$ elements with 
    cubic polynomials, the number of degrees of freedom for 
    $\fvor \in \Ltwo{p}$ 
    is 147456. 
    For $\fpot \in \Hone{p}$, the number of degrees of freedom is 82944. 
    Temporal discretizations imp, bwe, emp, and RK4 are stable for any $\delta$
    while fwe requires a~full upwinding $\delta = 0.5$ to finish the~simulation.
    These simulations were obtained using 72 CPUs.}
  \label{tab:wallClockTime}
\end{table}

\subsection{Forced 2D Turbulence \label{sec:forced_turbulence} }
\begin{figure}[t]
  \centering
  \includegraphics[width=\linewidth]{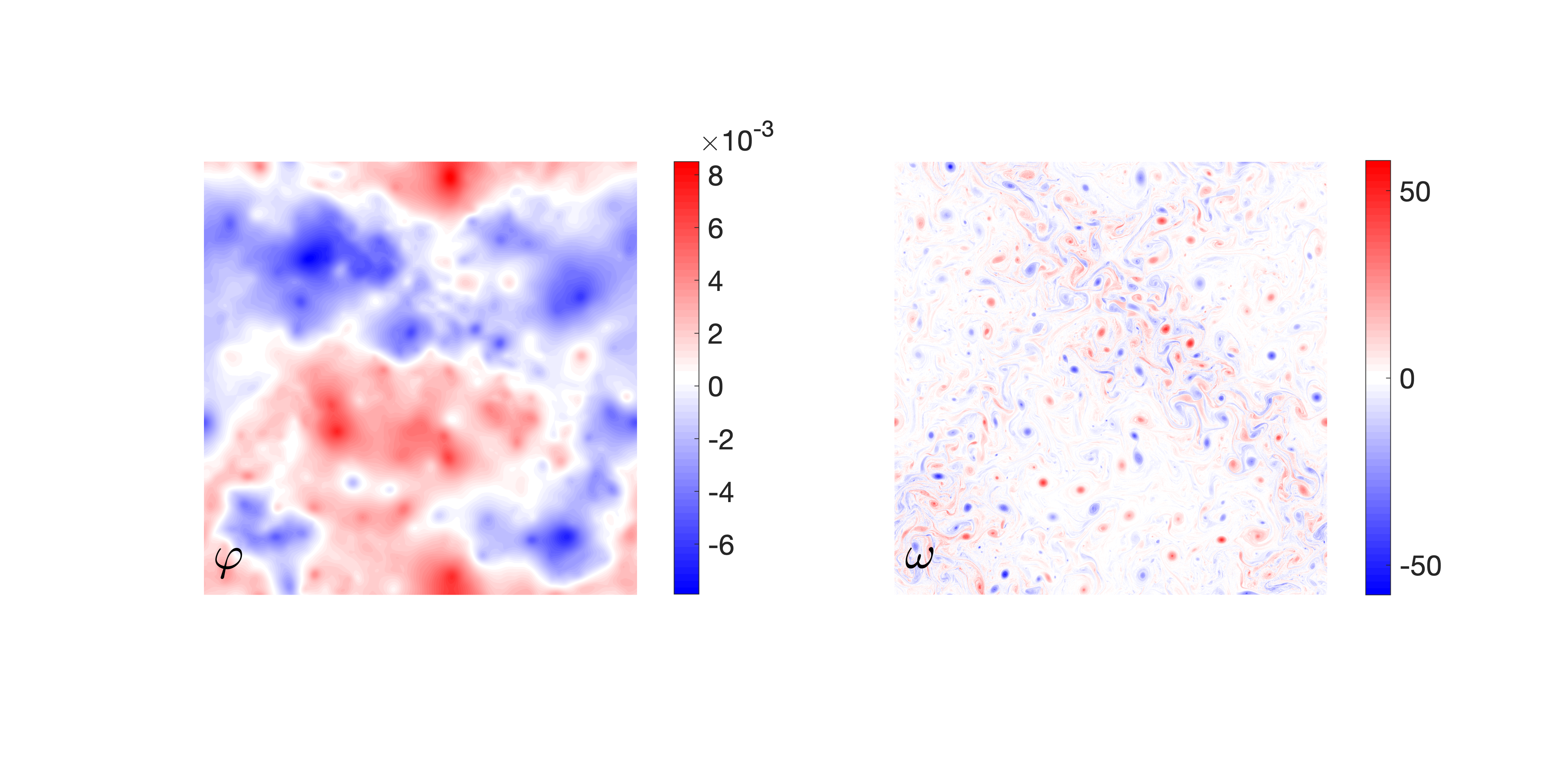}%
  \caption{Snapshots of potential and vorticity during forced 2D turbulence MFEM simulations.}%
   \label{fig:skh_snapshot}
\end{figure}

To further quantify the improvement of the energy and enstrophy conserving scheme, and to benchmark with other codes, we now include physical dissipation and the  vorticity source $S_\fvor$, i.e. Eq.~\eqref{eq:NS_vorticity}, in the model.
This model describes saturated forced 2D turbulence and has been studied extensively, e.g. in \cite{Boffetta2012arfm}.
 
\subsubsection{Simulation setup}
Although many studies have been performed with random forcing functions, e.g. that are white noise in time, which reduces the time necessary to accumulate uncorrelated statistics, here a source that is more easily made deterministic is used to ensure that the results are reproducible to the greatest possible extent.
The deterministic stirring force is chosen that has enough spectral content to 
excite significant flows at all scales
\begin{multline}
    S_\fvor=s_0 \sum_{m=m_0}^{m_0+\Delta m}
     \cos{\left(2\pi m\left(x - y\right)-\phi_m + \omega_0 t\right)}
    -\cos{\left(2\pi m\left(x + y\right)-\phi_m - \omega_0 t\right)}
    \\
    +\sin{\left(2\pi m\left(x - y\right)-\phi_m + \beta \omega_0 t\right)}
    - \sin{\left(2\pi m\left(x+ y\right)-\phi_m - \beta\omega_0 t\right)}
\end{multline}
with phase shift $\phi_m=( (m-m_0)\, {\rm mod}\,5) 2\pi/5$ and frequency ratio $\beta=2/(1+\sqrt{5}) \simeq 1/1.618$.
A square simulation domain with $L_x=L_y=1$ is chosen for this study. 

In the following, the  ``non-conservative setup" refers to CG elements for both vorticity and potential and backwards Euler implicit time integration while  the ``conservative setup" refers to DG elements for the vorticity and implicit midpoint time integration. 
In this subsection, all the MFEM simulation results are obtained without upwinding, i.e. $\delta=0$,  unless noted otherwise.
In order to provide benchmark cases for our MFEM simulations, we also used the GDB code~\cite{Zhu2018cpc}, a finite difference code designed to study 3D tokamak edge turbulence~\cite{Zhu2017pop}. For this study, GDB is configured to run this particular 2D turbulence model. 
In particular, GDB employs an energy- and enstrophy-conserving Arakawa scheme~\cite{Arakawa1966jcp} for the advection operator, 4th order central differences for the spatial derivatives, the trapezoidal leapfrog method \cite{Zalesak1979jcp,Jardin2010book} for the time integration, which treats diffusion implicitly, and a spectral method to solve for  $\fpot$.

\begin{figure}[t]
  \centering
  \includegraphics[width=\linewidth]{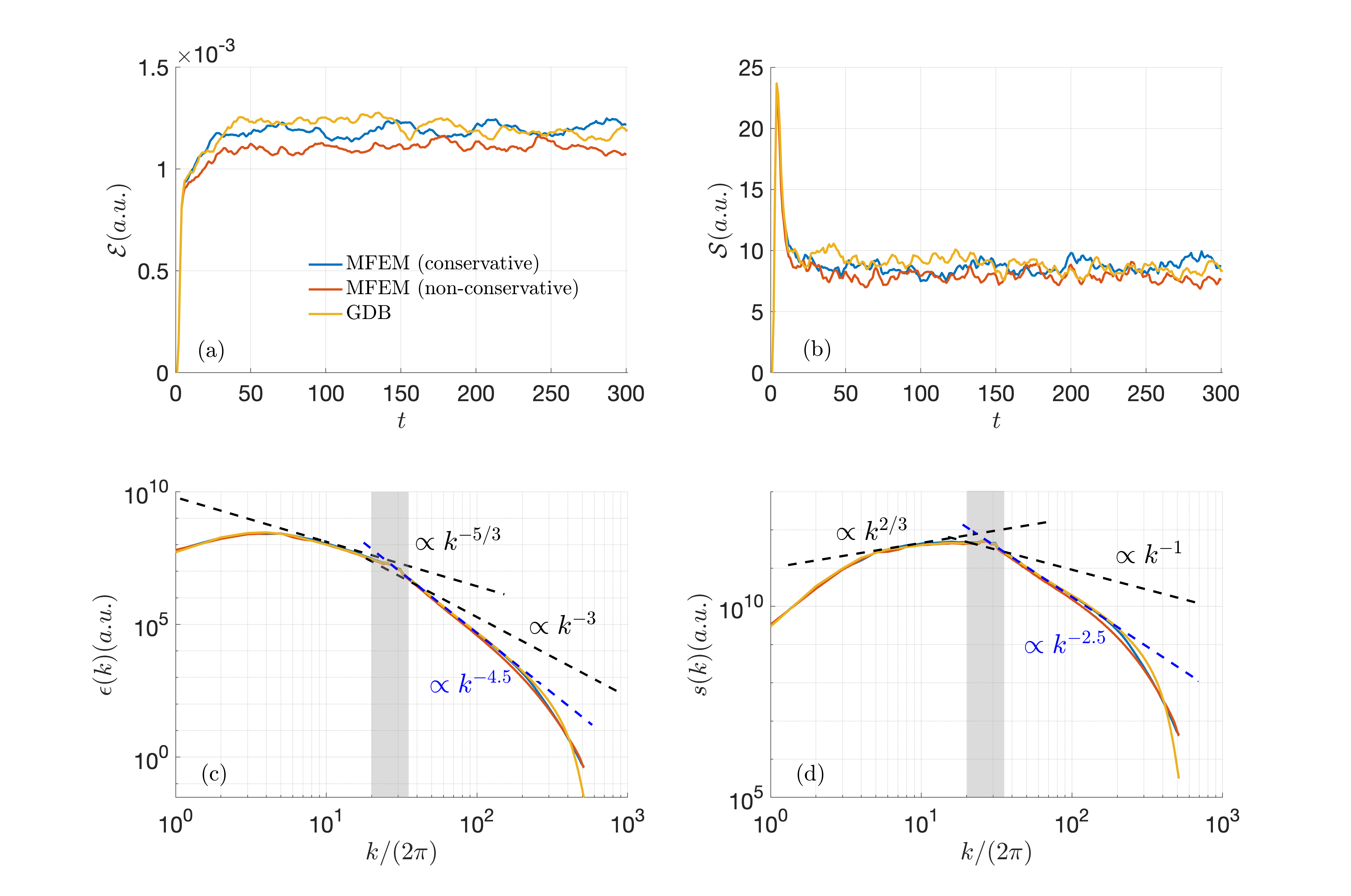}%
  \caption{Domain averaged (a) energy and (b) enstrophy evolution, and time-averaged (c) energy and (d) enstrophy spectra from non-conservative (H1 and backward Euler method), conservative (DG and second order implicit midpoint method) MFEM and reference GDB simulations for $\alpha=0.02,\mu=10^{-6}$.
  The spectral forcing region is shown in grey.}%
  \label{fig:skh_ee_dual}
\end{figure}

\subsubsection{Stationary spectrum}
The first study focuses on the dual cascade with parameters $\alpha=0.02$, $\mu=10^{-6}$, $s_0=1,m_0=18,$ $\Delta m=5$ (Benchmark 1). The solution is initialized with zero potential and vorticity fields in a fully periodic square domain.
To ensure that the MFEM and GDB runs have similar resolution and similar numbers of unknowns to solve for, different grid resolutions are used for GDB and MFEM.
For GDB, the grid resolution is $1024\times 1024$; while For MFEM, the grid resolution is $512\times 512$ with 2nd-order elements, i.e., polynomial order $p=2$. Therefore, the non-conservative MFEM setup has $1024^2=1048576$ degrees of freedom (DOFs) for both potential and vorticity, and the conservative MFEM setup has $1048576$ and $2359296$ DOFs for potential and vorticity respectively. The MFEM results are then interpolated to the GDB grid for further spectral analysis.

As illustrated in Fig.~\ref{fig:skh_snapshot}, after a brief period of instability growth, the system enters a nonlinear saturated turbulent stage. The Reynolds number is of the order $Re\sim 10^6$.
Figure \ref{fig:skh_ee_dual} (a) and (b) show the evolution of domain averaged energy and enstrophy in the system.
To evaluate the energy and enstrophy in this turbulent system, the time averaged $\avg{\Energy}$ and $\avg{\Enstrophy}$ are calculated over the interval $t\in[150,300]$, i.e., during the fully saturated turbulent stage, and tabulated in Table~\ref{tab:EE_benchmark}. 
Although both conservative and non-conservative MFEM and GDB runs have the same level of $\avg{\Energy}$ in this moderate ``dispersive" case, the non-conservative MFEM setup has roughly $10\%$ less $\avg{\Enstrophy}$ than the conservative MFEM and GDB results. Thus, the conservative MFEM setup does a better job of qualitatively matching the GDB result in terms of retaining a similar amount of energy and enstrophy in the system.
We also examine the simulated turbulence characteristics by comparing the energy and enstrophy spectral densities ($\epsilon(k)$ and $s(k)$) defined via
\begin{align}
  \epsilon(k) := d\Energy/dk =  -\oint \tfrac{1}{2} \langle \hat{\fpot}^* \hat{\fvor} \rangle \, kd\theta_k
  \\
   s(k) := d\Enstrophy/dk = \oint \tfrac{1}{2} \langle \hat{\fvor}^* \hat{\fvor} \rangle \, kd\theta_k
  .
\end{align}
where $\hat{f}$ is the Fourier representation of quantity $f$ in $k-$space and $\theta_k=\arctan{(k_y/k_x)}$ is the angle of the wavevector.
Both $\epsilon(k)$ and $s(k)$ are time-averaged over $t\in[150,300]$ and plotted in Fig. \ref{fig:skh_ee_dual} (c) and (d). 
A clear dual cascade of energy and enstrophy  from the forcing spectral region shown in grey to both longer and shorter wavelengths can be seen. The inverse cascade to long wavelengths (low wavenumbers) is anticipated to follow the Kolmogorov-Obukhov scaling laws: $d\Energy/dk \propto k^{-5/3}$ and $d\Enstrophy/dk \propto k^{1/3}$.
This scaling is well-reproduced by the data for both GDB and MFEM, despite the relatively small range in which the scaling can be observed. If enough of an inertial range was retained for the direct (forward) cascade to short wavelength (high wavenumbers), the results would be expected to follow the Kraichan scaling laws $d\Energy/dk \propto k^{-3}$ and $d\Enstrophy/dk \propto k^{-1}$. However, in this case, the viscosity is high enough that the direct cascade inertial range begins to blend into the dissipation range and the scaling is closer to $d\Energy/dk \propto k^{-4.5}$ and $d\Enstrophy/dk \propto k^{-2.5}$.
In general, the comparison between MFEM and GDB results is qualitatively good across the entire spectra perhaps with the exception at high wavenumber ($k/(2\pi)\geq 300$) where GDB result exhibits a slightly steeper scaling.
The substantial differences in the numerics for the advection operator as well as the spectral solve for the potential are possible explanations for this discrepancy.

The simulation results from the non-conservative MFEM setup deviate more from the conservative MFEM and GDB results at smaller $\mu$ because, as the physical dissipation becomes smaller, the numerical dissipation starts to play an important role.
Figure \ref{fig:skh_ee_dual_mu1e-8} shows the comparison between MFEM and GDB results with the same setup except that the viscosity is reduced by two orders of magnitude to $\mu=10^{-8}$ (Benchmark 2).
With these parameters, the reference GDB run yields $\avg{\Energy} = 1.525\times 10^{-3}$ and $\avg{\Enstrophy} = 22.21$. 
The conservative MFEM setup gives the similar numbers with $\avg{\Energy} = 1.525\times 10^{-3}$ and $\avg{\Enstrophy} = 21.32$; while for the non-conservative MFEM setup, $\avg{\Energy} = 1.311\times 10^{-3}$ and $\avg{\Enstrophy} = 10.59$ - roughly $13\%$ and $50\%$ less than the conservative MFEM and GDB results.
On the other hand, as now the physical viscosity is too small to effectively dissipate injected vorticity at high wavenumber, grid scale turbulent vortices produced by the forward cascade process are more apparent in the simulations, causing a ``pile-up" of the energy and, even more apparently, a ``pile-up" of enstrophy towards the high-$k$ end of the spectrum. 
This phenomenon can clearly be seen in Figs.~\ref{fig:skh_ee_dual_mu1e-8} (c) and (d).
A study of how to mitigate this pile-up will be presented in the next section.
Overall, the conservative MFEM and GDB spectra are in good agreement up to $k/(2\pi)\approx 100$ and only diverge at high-$k$.
As discussed above, this may due to the use of distinct numerical approaches.

\begin{figure}[h]
  \centering
  \includegraphics[width=\linewidth]{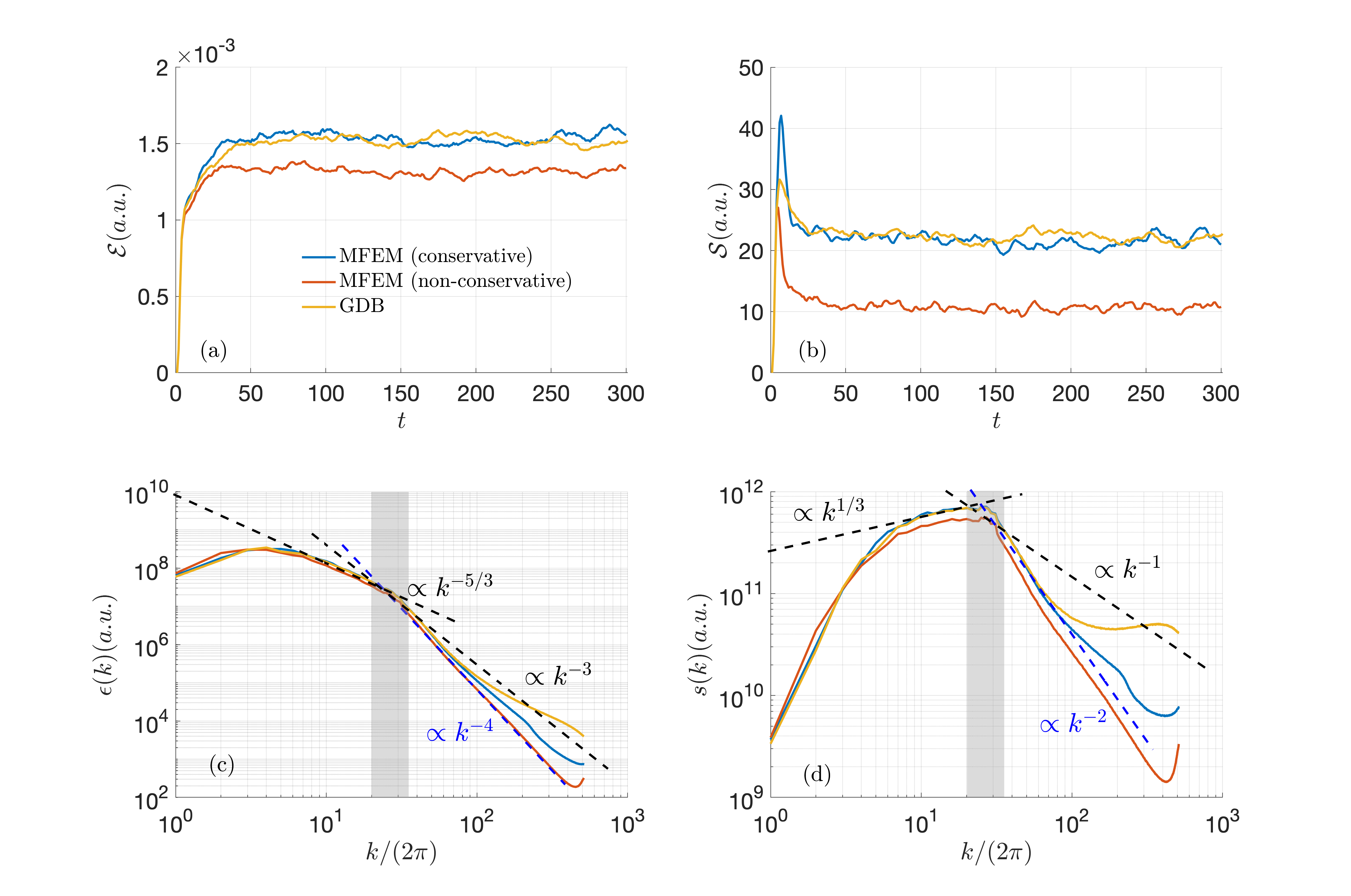}%
  \caption{Domain averaged (a) energy and (b) enstrophy evolution, and time-averaged (c) energy and (d) enstrophy spectra from conservative (DG and second order implicit midpoint method), non-conservative (H1 and backward Euler method) MFEM and reference GDB simulations for $\alpha=0.02,\mu=10^{-8}$.
  The spectral forcing region is shown in  grey.}%
  \label{fig:skh_ee_dual_mu1e-8}
\end{figure}

In this dual cascade study, $k^{-4.5}$ and $k^{-4}$ forward cascade energy scalings are observed for $\mu=10^{-6}$ and $\mu=10^{-8}$ respectively. We now focus more directly on the direct cascade region by shifting the spectrum of the forcing function to long wavelength with parameters $\alpha=0.02$, $\mu=10^{-8}$, $s_0=0.2,m_0=3,$ $\Delta m=5$ (Benchmark 3).
Other parameters, such as resolutions and polynomial order, are kept the same as the Benchmark 1 and 2. In this benchmark, stirring the system with larger coherent structures and the lower viscosity allow more of an inertial range to develop in the direct cascade region.
Hence, the resulting forward cascade scaling would be expected to be closer to the theoretical prediction.
As shown in Table~\ref{tab:EE_benchmark}, the conservative MFEM setup and the GDB code again produce qualitatively consistent $\avg{\Energy}$ and $\avg{\Enstrophy}$ with $<4\%$ discrepancy in this benchmark.
In contrast, the non-conservative MFEM setup has $\sim 30\%$ lower $\avg{\Enstrophy}$.
Figures \ref{fig:skh_ee_direct} (c) and (d) show the energy and enstrophy spectra.
The measured scaling laws are now qualitatively similar to those predicted from Kraichnan's theory: $d\Energy/dk\propto k^{-3}$ and $d\Enstrophy/dk\propto k^{-1}$ for both conservative MFEM and GDB results.
Simulations that retain more of an inertial range (i.e., higher resolution) would provide even better confirmation of this.
The energy and enstrophy spectra from the conservative MFEM and GDB codes are once again in good agreement, except at very high wavenumbers where the differences are likely attributable to differences in the numerical algorithms.

From these physically important benchmarks, we conclude that the conservative MFEM outperforms the non-conservative MFEM in both retaining a higher level of energy and enstrophy, especially when the system is nearly inviscid.
Moreover, its performance is comparable with the widely-used Arakawa scheme~\cite{Arakawa1966jcp} which is considered to be the gold standard of energy and enstrophy conserving discretizations of the advection operator within the finite-difference world.

\begin{figure}[t]
  \centering
  \includegraphics[width=\linewidth]{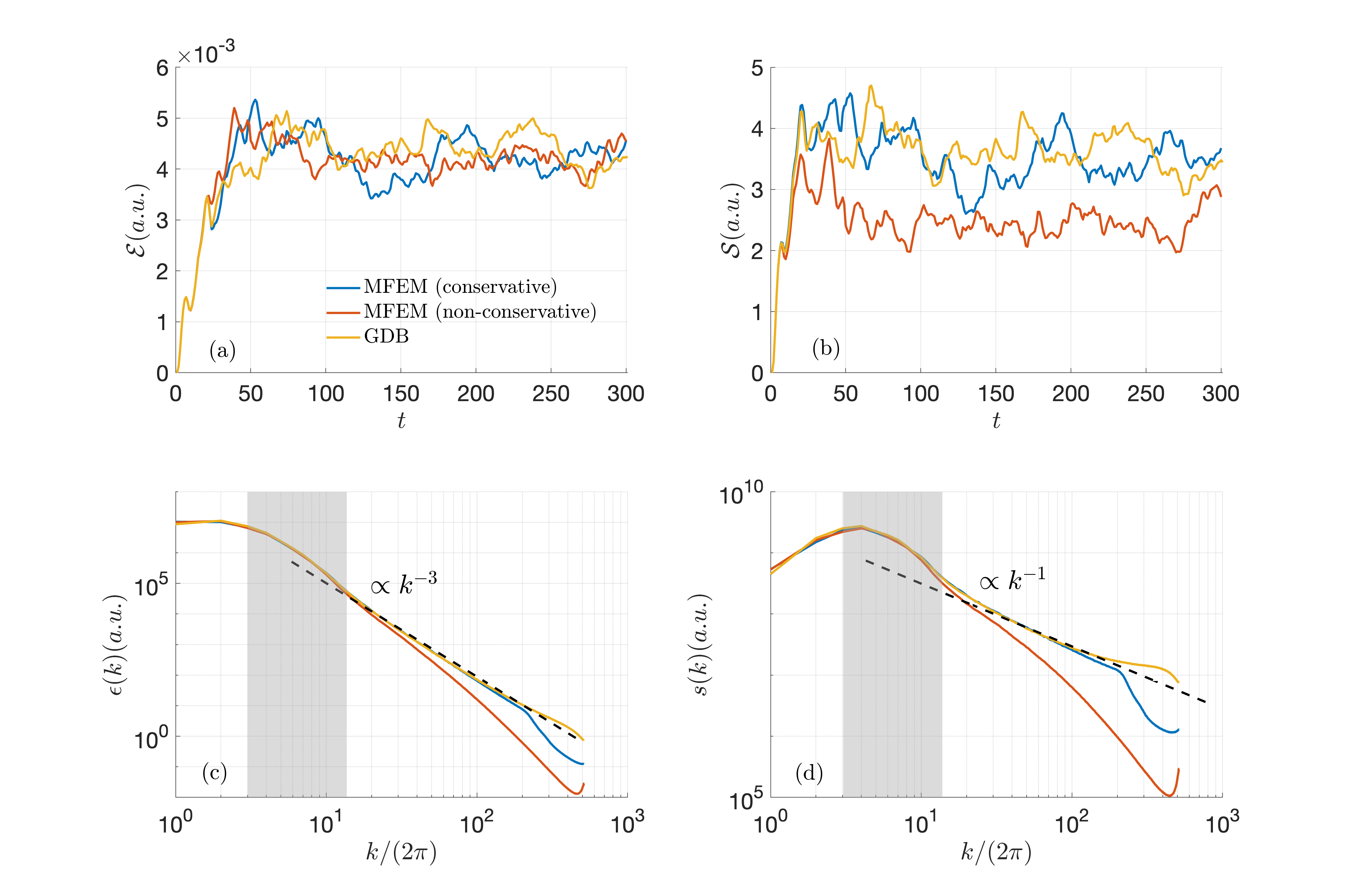}%
    \caption{
    Domain averaged (a) energy and (b) enstrophy evolution for MFEM and (c) energy and (d) enstrophy spectra from conservative (DG and second order implicit midpoint method), non-conservative (H1 and backward Euler method) MFEM and reference GDB simulations ($\alpha=0.02,\mu=10^{-8}$).
     The spectral forcing region is shown in  grey. }%
   \label{fig:skh_ee_direct}
\end{figure}
 
\begin{table}
\centering
\begin{tabular}{ c c c c c }
	\multirow{2}{*}{Benchmark} &  & MFEM & MFEM & \multirow{2}{*}{GDB} \\
			& & (conservative) & (non-conservative) & \\
	\hline 
	\multirow{2}{*}{1 (Fig.~\ref{fig:skh_ee_dual})} 
			& $\avg{\Energy}/10^{-3}$	   & 1.198 & 1.112 & 1.186 \\ 
			& $\avg{\Enstrophy}$            & 8.733 & 7.867 & 8.497 \\ 
	\hline 
	\multirow{2}{*}{2 (Fig.~\ref{fig:skh_ee_dual_mu1e-8})}
			& $\avg{\Energy}/10^{-3}$	   & 1.525 & 1.311 & 1.525 \\ 
			& $\avg{\Enstrophy}$            & 21.32 & 10.59 & 22.21 \\ 
	\hline 
	\multirow{2}{*}{3 (Fig.~\ref{fig:skh_ee_direct})}
			& $\avg{\Energy}/10^{-3}$	   & 4.231 & 4.171 & 4.408\\ 
			& $\avg{\Enstrophy}$            & 3.485 & 2.443 & 3.599\\ 
\end{tabular}
    \caption{Energy and enstrophy ($\avg{\Energy}$ and $\avg{\Enstrophy}$) averaged over $t\in[150,300]$ from the MFEM and GDB codes for three different sets of parameters.}
	\label{tab:EE_benchmark}
\end{table}

\begin{table}
\centering
\begin{tabular}{ c c c c c c c}
	Test & &\multicolumn{5}{c}{MFEM (conservative)} \\ \hline
	\multirow{5}{*}{1 (Fig.~\ref{fig:skh_ee_dual_ref})}
			& grid resolution	& $256^2$ & $512^2$ & $1024^2$ \\
			& DOFs$_\fpot$ & $512^2$ & $1024^2$ & $2048^2$ \\ 
			& time factor & 1.0 & 2.371 & 14.343 \\ \cline{2-5}
			& $\avg{\Energy}/10^{-3}$	   & 1.531 & 1.525 & 1.509 \\ 
			& $\avg{\Enstrophy}$            & 39.80 & 21.32 & 16.29 \\ 
	\hline 
	\multirow{5}{*}{2 (Fig.~\ref{fig:skh_ee_dual_order})}
			& $p$-order & 2 & 3 & 4 \\ 
			& DOFs$_\fpot$ & $512^2$ & $510^2$ & $512^2$ \\ 
			& time factor & 1.0 & 1.857 & 2.334 \\ \cline{2-5}
			& $\avg{\Energy}/10^{-3}$	   & 1.554 & 1.588 & 1.507 \\ 
			& $\avg{\Enstrophy}$            & 40.42 & 39.28 & 34.74 \\ 
	\hline 
	\multirow{4}{*}{3 (Fig.~\ref{fig:skh_ee_dual_upwind})}
			& $\delta$ & 0 & 0.0005 & 0.005 & 0.05 & 0.5 \\ 
			& time factor & 1.0 & 0.982 & 0.932 & 0.520 & 0.481 \\  \cline{2-7}
			& $\avg{\Energy}/10^{-3}$	   & 1.572 & 1.525 & 1.564 & 1.531 & 1.496\\ 
			& $\avg{\Enstrophy}$            & 183.2 & 37.84 & 17.29 & 15.00 & 14.20\\ 
\end{tabular}
    \caption{Energy and enstrophy ($\avg{\Energy}$ and $\avg{\Enstrophy}$) averaged over $t\in[150,300]$ from the conservative MFEM simulations with different setups.} 
	\label{tab:EE_mfem_midk}
\end{table}

\begin{figure}[h]
  \centering
  \includegraphics[width=\linewidth]{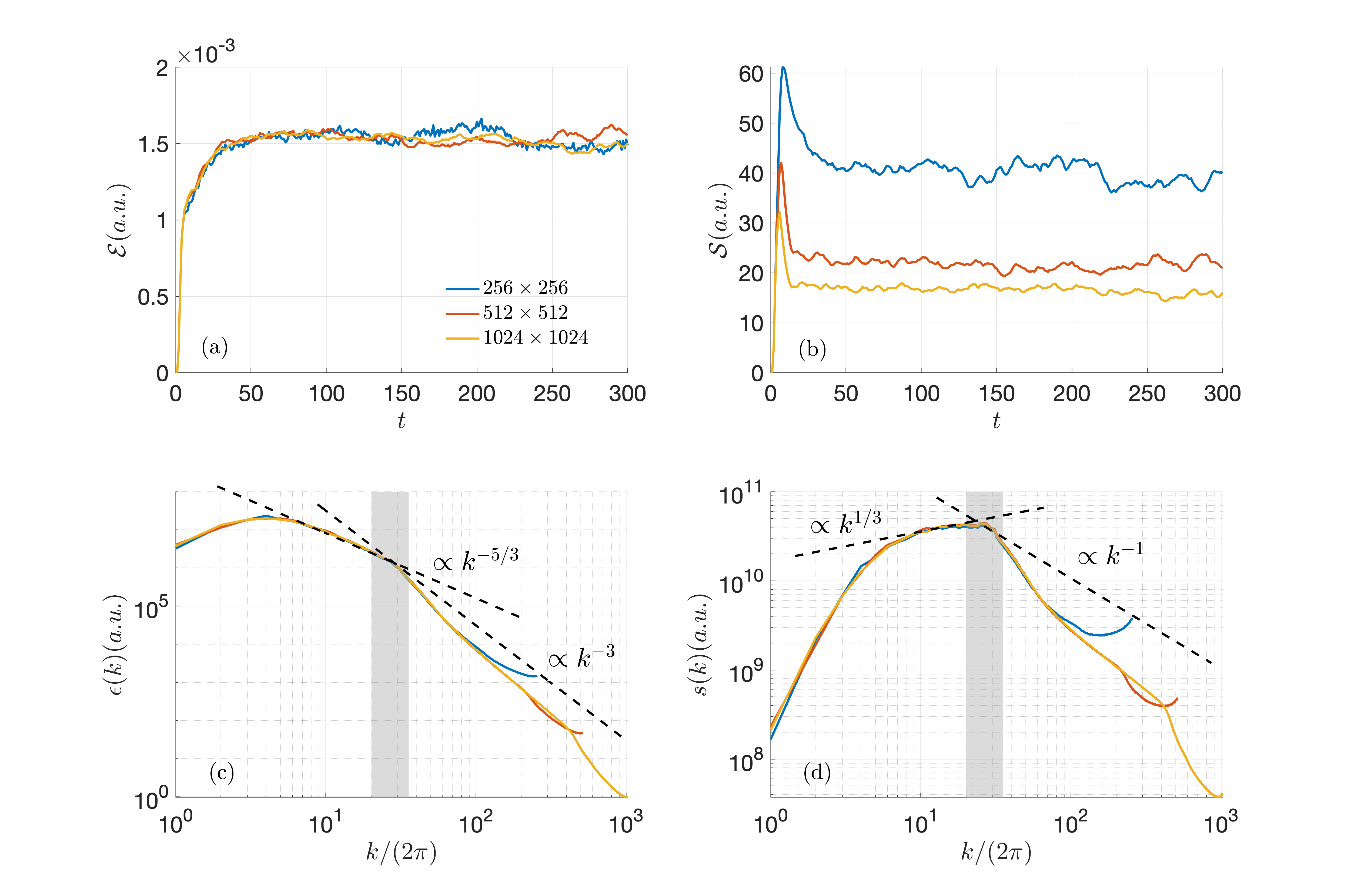}%
  \caption{ Domain averaged (a) energy and (b) enstrophy evolution, and (c) energy and (d) enstrophy spectra from conservative (DG and second order implicit midpoint method) MFEM simulations with second order polynomial but different grid resolutions ($\alpha=0.02,\mu=10^{-8}$).
  The spectral forcing region is shown in  grey.}%
  \label{fig:skh_ee_dual_ref}
\end{figure}

\begin{figure}[h]
  \centering
  \includegraphics[width=\linewidth]{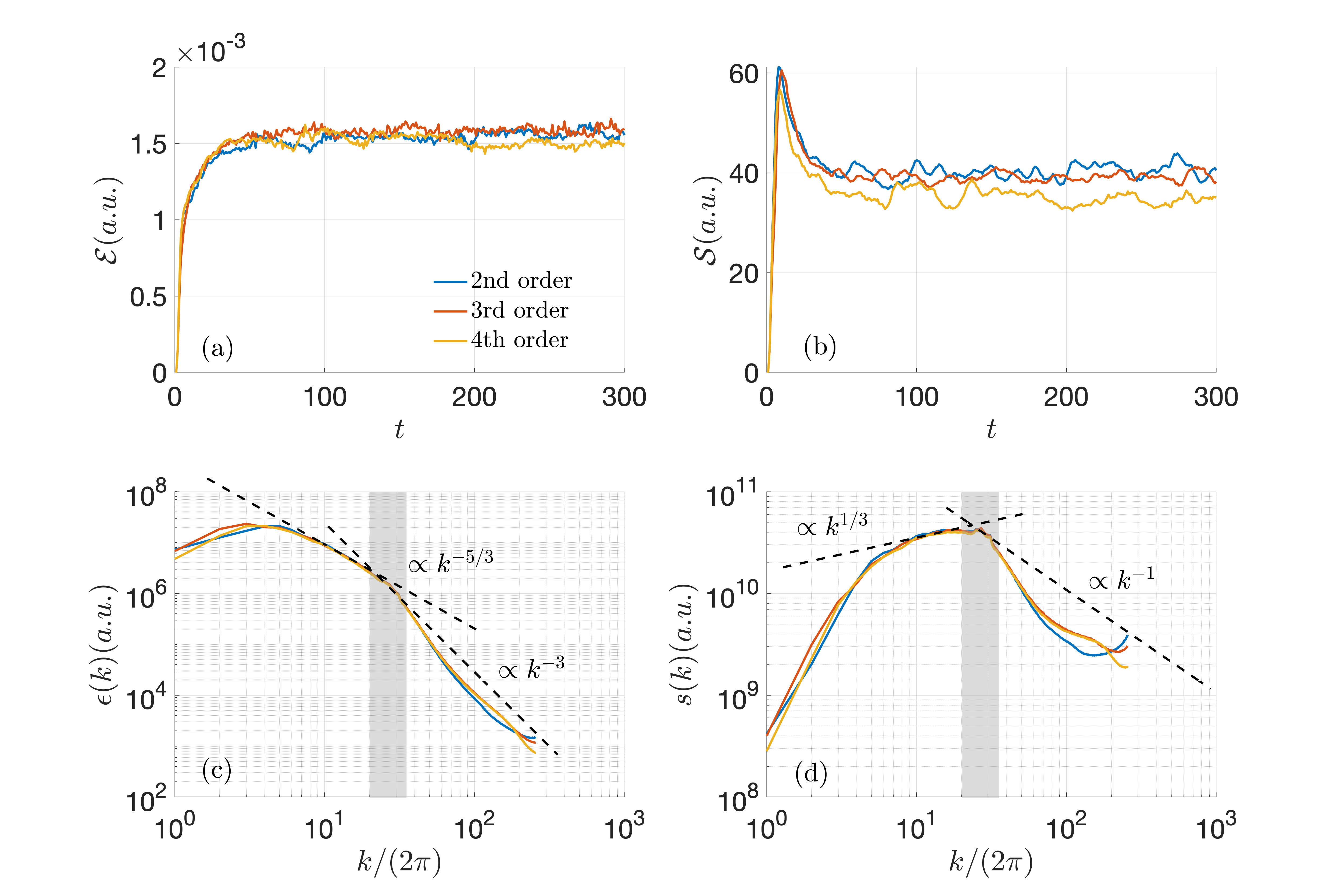}%
  \caption{Domain averaged (a) energy and (b) enstrophy evolution, and (c) energy and (d) enstrophy spectra from conservative (DG and second order implicit midpoint method) MFEM simulations with different  polynomial order $p$ but roughly the same DOFs on $fpot$ ($\alpha=0.02,\mu=10^{-8}$).
  The spectral forcing region is shown in  grey.}%
  \label{fig:skh_ee_dual_order}
\end{figure}

\begin{figure}[h]
  \centering
  \includegraphics[width=\linewidth]{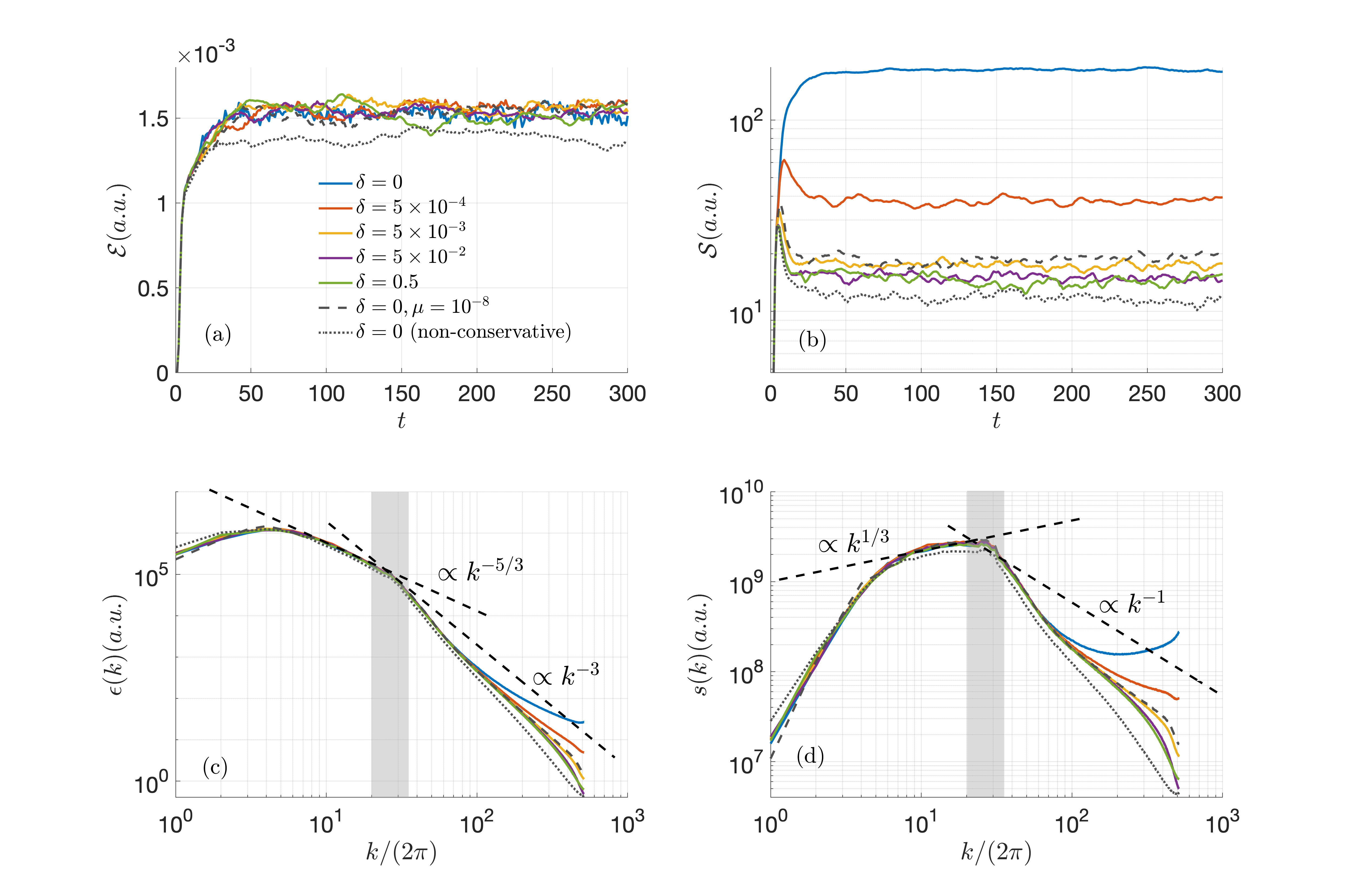}%
  \caption{Domain averaged (a) energy and (b) enstrophy evolution, and (c) energy and (d) enstrophy spectra from the MFEM simulations with different level of upwinding $\delta$. Here, $\alpha=0.02, \mu=0$ for all cases except the grey dashed lines with $\mu=10^{-8}$ for the grey dashed line. All results expect the grey dotted lines are from conservative MFEM runs.
  The spectral forcing region is shown in grey.}%
  \label{fig:skh_ee_dual_upwind}
\end{figure}

\subsubsection{Controlling the high-$k$ spectrum}
The high-$k$ pile-up that appears in our FEEC model (e.g., in Figs. \ref{fig:skh_ee_dual_mu1e-8} (c) and (d)) due to insufficient dissipation can be addressed in a few different ways: by increasing the~mesh resolution, by changing 
the polynomial order $p$, and/or by applying upwinding (i.e., $\delta > 0$). To illustrate this, three sets of tests were performed with the conservative MFEM setup under parameters $\alpha=0.02$, $s_0=1,m_0=18,$ $\Delta m=5$ and:
\begin{enumerate}
\item varying resolution for $\mu=10^{-8}$ and $p=2$; \\
\item varying polynomial order $p$ but keep roughly the same DOFs on $fpot$ with $\mu=10^{-8}$; \\
\item applying upwinding level $\delta$ on a $512\times 512$ grid and $\mu=0$, $p=3$.
\end{enumerate}
Figures \ref{fig:skh_ee_dual_ref}, \ref{fig:skh_ee_dual_order} and \ref{fig:skh_ee_dual_upwind} plot the corresponding test results, and the averaged $\avg{\Energy}$ and $\avg{\Enstrophy}$ of all the simulations are summarized in Table~\ref{tab:EE_mfem_midk}.

As shown in Figure \ref{fig:skh_ee_dual_ref}(a) and (b) and Table~\ref{tab:EE_mfem_midk}, increasing resolution from $256\times 256$ to $1024\times 1024$ has little influence on $\avg{\Energy}$ (decreasing less than $2\%$) since the majority of energy is carried by low-$k$ vortices that have been well-resolved even with the relatively coarse grid used in this case. Nonetheless, increasing resolution reduces $\avg{\Enstrophy}$ substantially from $39.80$ to $16.29$. This is because, with higher resolution, smaller vortices are now captured on the finer grid and more easily dissipated by the physical viscosity. 
As a result, not only has the high-$k$ pile-up been eliminated, but the inertial range is extended, and the forward cascade scalings are closer to the theoretical predictions, as shown in Fig. \ref{fig:skh_ee_dual_ref}(c) and (d). The main drawback of this approach is the demanding computational cost. In our test, the high grid resolution $1024\times 1024$ run requires about 14 times computational resources than the low grid resolution $256\times 256$ run.

The polynomial order $p$ can also impact the high-$k$ spectra. Because changing $p$ on a fixed grid would also change the number of DOFs, the simulated results in Fig. \ref{fig:skh_ee_dual_order} are carried out with reduced grid resolution as $p$ increases such that the DOFs of potential $\fpot$ are roughly constant. In Test 2, $\avg{\Energy}$ shows little variation with $p$, while $\avg{\Enstrophy}$ decreases as $p$ increases despite the DOFs$_\fpot$ remains the same. More interestingly, comparing to the second order polynomial results, the high-$k$ pile-up in enstrophy spectrum fades away for $p=3$ and is completely eliminated when the fourth order polynomial is applied. The computational cost has a modest increment from $p=2$ to $p=4$. 

Finally, upwinding is extremely effective at preventing the high-$k$ pile-up. 
For a simulation that is poorly resolved, it would likely be prudent to use the maximum upwinding parameter $\delta=1/2$, as is customary in the finite element community.
For these simulations, which are moderately under-resolved, it is interesting to explore the effect of the choice of $\delta$ on the results.
In the inviscid $\mu=0$ limit without upwinding ($\delta=0$), the enstrophy can only be dissipated through friction; therefore, $\avg{\Enstrophy}$ is maintained at a much higher level ($183.2$) and $s(k)$ has a significant high-$k$ pile-up, seen in Figs.~\ref{fig:skh_ee_dual_upwind}~(b) and (d) (blue lines). However, even with a small amount of upwinding, e.g., $\delta=5~\times~10^{-4}$, $\avg{\Enstrophy}$ is reduced by nearly five times to $37.84$ and the high-$k$ pile-up in enstrophy spectra is completely avoided. With $\delta=5\times 10^{-3}$, the  result is similar to a simulation with viscosity $\mu=10^{-8}$ but without upwinding (dashed grey lines) in Fig.~\ref{fig:skh_ee_dual_upwind}. 
For the forced turbulence model, upwinding mostly affects the high $k$-vortices near the grid-scale; increasing the upwinding level $\delta$ primarily affects the high-$k/(2\pi)>200$ spectra. 
As high-$k$ vortices contribute little to the overall energy, upwinding has limited impact on $\avg{\Energy}$ even when $\delta$ is set to the maximum value of $1/2$.
Notably, for this model, there is little change in the spectrum as $\delta$ ranges from 0.05 to 0.5; even for $\delta=0.5$ which is a common choice for most DG methods, the system is able to attain a significantly higher level of energy and enstrophy as well as reasonable accurate spectra comparing to the non-conservative MFEM results (grey dotted lines).
Another interesting feature of upwinding is that the computational cost reduces as the upwinding factor $\delta$ increases. This is because smoother solutions can be reached with fewer iterations, i.e,. easier to converge.

\section{Conclusion}
 
In conclusion, a finite element exterior calculus (FEEC) approach to the two-dimensional (2D) incompressible Navier-Stokes equations and to drift-reduced magnetohydrdynamics is introduced that conserves the two well-known quadratic invariants, energy and enstrophy, for arbitrary polynomial order.
Because the theorems of vector calculus are satisfied at the discrete level within the FEEC framework, the spatial discretization automatically conserves both energy and enstrophy as long as the proper finite element spaces are chosen for the primary variables.
A conservative scheme  also requires an appropriate choice of numerical time integrator, and, for B-series time integrators, only generally symplectic time integrators, such as the fully implicit midpoint method are conservative.

Various combinations of numerical methods were tested for an extensive suite of 2D neutral fluid turbulence test problems.  
Linear results were verified against a semi-analytical numerical eigensolver, and nonlinear results were verified against theory as well as to the Global Drift Ballooning (GDB) code \cite{Zhu2018cpc}.  
In addition to implicit midpoint, it was found that certain explicit  methods, such as the explicit midpoint method and the explicit trapezoidal leapfrog method, which approximate implicit midpoint, still  do relatively well at approximately conserving energy and enstrophy. 
These observations may be due to the relatively small time step sizes that must be taken for the explicit methods to remain stable. 
However, it was also found that turbulent simulations that conserve both energy and enstrophy tend to result in too much power at high wavenumber. 

The high wavenumber part of the spectrum can be controlled by reintroducing artificial dissipation, e.g. through numerical viscosity or through an asymmetric implicit time integrator, such as backward Euler. 
Perhaps the most interesting option is to use an upwinded DG formulation of the advection operator which dissipates enstrophy while still conserving energy exactly. 
This formulation makes the best compromise of allowing mid-range wavenumbers to reach the appropriate amplitude, while still controlling the high-wavenumber part of the spectrum.
Thus, our conclusion is that using the upwinded DG formulation in conjunction with the symplectic implicit midpoint method appears to offer the greatest advantages in terms of accuracy and efficiency.

\appendix
\section{DG formulation for the Advection Operator \label{app:DGformulation}}
The DG flux formulation of the vorticity Eq. \eqref{eq:vorticity} without sources and dissipation is
\begin{gather*}
  \int_{\Omega} \ftestLtwo \pdv{\fvor}{t} dV 
  - \int_\Omega \fvor \vDrift \cdot \Vnabla \ftestLtwo ~dV
  + \sum_{K\in \tau_h} \int_{\partial K} 
    \ftestLtwo \hat{\fvor} \vDrift \cdot \vec{n}_{\partial K} ~dS
  = 0, \forall \ftestLtwo \in \Ltwo{\qvor},
\end{gather*}
where  $\hat\fvor$ is an approximation of $\fvor$  that defines the flux $\hat{\fvor}\vDrift$ across $\Gamma$.
Using  the fact that
\begin{gather*}
\sum_{K\in \tau_h} \int_{\partial K} 
\ftestLtwo \hat{\fvor} \vDrift \cdot \vec{n}_{\partial K} ~dS = \int_{\Gamma}  
\jumpOp{ \ftestLtwo \hat{\fvor} } \vDrift \cdot \vec{n}_\Gamma~dS
+ \int_{\partial \Omega} \ftestLtwo \hat{\fvor} \vDrift \cdot 
\vec{n}_{\partial \Omega} ~dS
\end{gather*}
 and Eq. \eqref{eq:divTheoremDGproduct}, the vorticity equation becomes
\begin{multline*}
  \int_{\Omega} \ftestLtwo \pdv{\fvor}{t} dV 
  - \int_\Omega \fvor \vDrift \cdot \Vnabla \ftestLtwo ~dV
  + \int_{\partial\Omega} \fvor \ftestLtwo 
    \vDrift \cdot \vec{n}_{\partial \Omega} ~dS
  \\
  + \int_\Gamma ( \jumpOp{ \hat{\fvor}  } \avgOp{ \ftestLtwo  } 
  + \avgOp{ \hat{\fvor} } \jumpOp{ \ftestLtwo  } ) 
  \vDrift \cdot \vec{n}_\Gamma dS = 0, 
  \hspace{0.25in} 
  \forall \ftestLtwo \in \Ltwo{\qvor} .
\end{multline*}

The conservative choice for the numerical flux 
\begin{gather}
\hat{\fvor} = \avgOp{ \fvor } + \delta 
\frac{|\vDrift\cdot\vec{n}_\Gamma|}{\vDrift\cdot\vec{n}_\Gamma}
\jumpOp{ \fvor }
\end{gather}
clearly has a vanishing jump, $\jumpOp{\hat\fvor} =0$, and satisfies $\avgOp{\hat\fvor}=\hat\fvor$ on $\Gamma$. 
Hence, this leads to 
\begin{multline*}
  \int_{\Omega} \ftestLtwo \pdv{\fvor}{t} dV 
  - \int_\Omega \fvor \vDrift \cdot \Vnabla \ftestLtwo ~dV
  + \int_{\partial\Omega} \fvor \ftestLtwo 
    \vDrift \cdot \vec{n}_{\partial \Omega} ~dS
    \\  
  + \int_\Gamma  
  \avgOp{ \fvor } \jumpOp{ \ftestLtwo } \vDrift \cdot \vec{n}_\Gamma dS 
  + \delta \int_\Gamma |\vDrift\cdot\vec{n}_\Gamma|
  \jumpOp{ \fvor }  \jumpOp{ \ftestLtwo } dS 
  = 0,
   \hspace{0.25in} 
    \forall \ftestLtwo \in \Ltwo{\qvor}.
\end{multline*}
Again using Eq. \eqref{eq:divTheoremDGproduct} to integrate by parts allows one to show that the 
weak form of the advection operator can be written in two equivalent ways
\begin{multline}
  \int_{\partial\Omega} \fvor \ftestLtwo 
    \vDrift \cdot \vec{n}_{\partial \Omega} ~dS -  \int_\Omega \fvor  \vDrift \cdot \Vnabla \ftestLtwo   ~dV
  + \int_\Gamma  
  \avgOp{ \fvor} \jumpOp{ \ftestLtwo } \vDrift \cdot \vec{n}_\Gamma dS 
  + \delta \int_\Gamma |\vDrift\cdot\vec{n}_\Gamma|
  \jumpOp{ \fvor }  \jumpOp{ \ftestLtwo } dS 
  = 
  \\
  \int_\Omega \ftestLtwo  \vDrift \cdot \Vnabla \fvor  ~dV
  - \int_\Gamma  
  \avgOp{ \ftestLtwo} \jumpOp{ \fvor } \vDrift \cdot \vec{n}_\Gamma dS 
  + \delta \int_\Gamma |\vDrift\cdot\vec{n}_\Gamma|
  \jumpOp{ \fvor }  \jumpOp{ \ftestLtwo } dS 
  ,
   \hspace{0.25in} 
    \forall \ftestLtwo \in \Ltwo{\qvor},
\end{multline}
This leads to the equivalent weak form of the vorticity equation given in Eq. \eqref{eq:vorticityIdealWeakGen}:
\begin{multline*}
  \int_{\Omega} \ftestLtwo \pdv{\fvor}{t} dV 
  + \int_\Omega \ftestLtwo  \vDrift \cdot \Vnabla \fvor  ~dV
     \\  
  - \int_\Gamma  
  \avgOp{ \ftestLtwo} \jumpOp{ \fvor } \vDrift \cdot \vec{n}_\Gamma dS 
  + \delta \int_\Gamma |\vDrift\cdot\vec{n}_\Gamma|
  \jumpOp{ \fvor }  \jumpOp{ \ftestLtwo } dS 
  = 0,
   \hspace{0.25in} 
    \forall \ftestLtwo \in \Ltwo{\qvor}.
\end{multline*}

\section*{Acknowledgements}
The authors would like to thank Benjamin Dudson, Alex Friedman, Tzanio Kolev, Mark Stowell, Mayya Tokman, and Xueqiao Xu for a number of valuable discussions.
This work was performed by LLNL under the auspices of the U.S. DOE under contract DE-AC52-07NA2734, and was supported by LLNL Laboratory Directed Research and Development project PLS-20-ERD-038. BSS was supported as a Nicholas C. Metropolis Fellow under the Laboratory Directed
Research and Development program of Los Alamos National Laboratory. Los Alamos National Laboratory report number LA-UR-22-21057. Release number LLNL-JRNL-831096.
\bibliography{mybibfile}

\end{document}